\documentclass[12pt]{article}

\usepackage{graphicx}
\usepackage{color,xcolor}
\usepackage{amsthm,amsmath,indentfirst,amssymb}
\usepackage{algorithm,float}
\usepackage{algorithmic}
\usepackage{hyperref}
\usepackage{url}
\usepackage{cite}
\usepackage{tikz}
\usepackage{subfigure}
\usepackage{pgfplots}

\tikzset{elegant/.style={smooth,thick,samples=50,magenta}}

\usepackage[top=1in, bottom=1in, left=1in, right=1in]{geometry}
\usepackage{setspace}

\newtheorem{lemma}{Lemma}[section]
\newtheorem{theorem}[lemma]{Theorem}

\newtheorem{coro}[lemma]{Corollary}

\newtheorem{definition}[lemma]{Definition}

\date{}

\begin{document}
	\title{Approximation Algorithm for Unrooted Prize-Collecting Forest with Multiple Components and Its Application on  Prize-Collecting Sweep Coverage}
	\author{Wei Liang\textsuperscript{1}\thanks{Email: lvecho1019@zjnu.edu.cn}, Shaojie Tang\textsuperscript{2}\thanks{Email: shaojie.tang@
			utdallas.edu}, Zhao Zhang\textsuperscript{1}\thanks{Corresponding author: Zhao Zhang, hxhzz@sina.com}\\
		\small \textsuperscript{1} School of Mathematical Sciences, Zhejiang Normal University\\
		\small Jinhua, Zhejiang, 321004, People's Republic of China\\
		\small \textsuperscript{2} Naveen Jindal School of Management, University of Texas at Dallas\\
		\small Richardson, Texas 75080, USA
	 }

\maketitle
\begin{abstract}
	
In this paper, we introduce a polynomial-time 2-approximation algorithm for the Unrooted Prize-Collecting Forest with $K$ Components (URPCF$_K$) problem. URPCF$_K$ aims to find a forest with exactly $K$ connected components while minimizing both the forest's weight and the penalties incurred by unspanned vertices. Unlike the rooted version RPCF$_K$, where a 2-approximation algorithm exists, solving the unrooted version by guessing roots leads to exponential time complexity for non-constant $K$. To address this challenge, we propose a rootless growing and rootless pruning algorithm. We also apply this algorithm to improve the approximation ratio for the Prize-Collecting Min-Sensor Sweep Cover problem (PCMinSSC) from 8 to 5.


\end{abstract}


\section{Introduction}

The Prize-Collecting Steiner Tree (PCST) problem, a classic combinatorial optimization problem, has received significant attention in the literature (refer to surveys \cite{PCSTsurvey,STsurvey}). lt plays an important role in various domains, including network design \cite{Johnson2000,Leitner2018}, vehicle routing \cite{Harks2013,liang2023,Tang2022}, and more. In the PCST problem, we are given a graph with an edge cost function represented by $w$ and a vertex penalty function denoted as $\pi$. The objective is to find a tree $T$ that minimizes both the sum of edge cost of $T$ and the total penalty associated with vertices not covered by $T$. PCST is known to be NP-hard and serves as a generalization of the Steiner tree problem \cite{karp1972}. A 2-approximation algorithm was proposed \cite{Goeman1995,feofiloff2007} for this problem, and later, Archer et al. \cite{archer2011} improved the approximation ratio to 1.9672.


In this paper, we study an extension of the problem known as the Prize-Collecting Forest with $K$ Components (PCF$_K$) problem. When provided with a positive integer $K$, PCF$_K$ seeks to identify a forest, denoted as $F$, comprising exactly $K$ components while minimizing the sum of the cost of the forest's edges and the penalties incurred by vertices not spanned by $F$.


While various PCST variants exist (see Related Work), it is notable that PCF$_K$, a natural extension, received limited attention until 2020 when Liang {\it et al.} \cite{Liang-PCMSSC-TCS} introduced it as a foundation for their study on Prize-Collecting Min-Sensor Sweep Cover (PCMinSSC). In their work, an algorithm for {\em rooted} PCF$_K$ (RPCF$_K$) played a crucial role for the PCMinSSC problem with an added base station assumption. They pointed out that when solving PCMinSSC without base station assumptions, an {\em unrooted} PCF$_K$ (URPCF$_K$) problem may be needed, and solving URPCF$_K$ posed a substantial theoretical challenge. Further details on this challenge are explained as follows.

In a sweep cover problem, mobile sensors must periodically visit points of interest (Pols). PCMinSSC aims to optimize sensor routes by minimizing both the number of active sensors and the total penalty incurred by unvisited Pols. A 5-approximation algorithm for a specific PCMinSSC variant was developed in \cite{Liang-PCMSSC-TCS}, where a constant number of base stations is present and each active sensor must pass through a base station. This algorithm relies on a 2-approximation algorithm for the rooted PCFK problem (RPCF$_K$), which requires predefined roots, with each component containing exactly one root. The introduction of base stations was driven by the inability to devise an approximation algorithm for the unrooted PCF$_K$ problem (URPCF$_K$) from a theoretical standpoint. In \cite{liang2023}, it was established that an 8-approximation algorithm for PCMinSSC could be achieved without assuming base stations. This approach employs a sub-procedure involving a 2-approximation algorithm for PCST on an auxiliary graph. The higher approximation ratio comes from transforming a tree in the auxiliary graph into a forest in the original graph. To improve the approximation ratio, an approach, as presented in this paper, is to directly explore approximation algorithms for URPCF$_K$.

The challenge caused by the {\em unrooted} concern lies in the following observations. For $K=1$, PCF$_1$ is exactly the prize-collecting Steiner tree problem PCST, and a rooted PCST algorithm $\mathcal A_{RPCST}$ can be used to solve the unrooted PCST problem as follows: run $|V|$ times of algorithm $\mathcal A_{RPCST}$, treating each vertex as a root, and choose the best one as a solution of the unrooted PCST instance. Indeed, many studies tackling the unrooted PCST problem have chosen to first convert it into its rooted counterpart, a practice observed in several papers\cite{Goeman1995,bienstock1993,Cole2001,archer2011}. This approach is favored due to the observation that the rooted PCST problem is more amenable to be formulated as an integer linear program.



This insight extends to situations where $K$ is a constant. When armed with an algorithm for the RPCF$_K$ problem, to solve the URPCF$_K$ problem, one can ``guess" $K$ roots, apply RPCF$_K$ for each conjecture, and subsequently select the most favorable outcome. By employing this approach, the URPCF$_K$ algorithm is executed $O(n^K)$ times, which is polynomial if $K$ is a constant.



A significant challenge arises when $K$ becomes an input, as seen in \cite{Liang-PCMSSC-TCS} when the authors address the PCMinSSC problem with base station assumptions. Their algorithm incorporates an algorithm for PCF$_K$ with $K$ varying from 1 to {the number of base stations}, which results in a running time exponentially depending on the number of base stations. This observation has prompted us to explore URPCF$_K$ directly.
	
In this paper, we have successfully developed a 2-approximation algorithm for URPCF$_K$, where $K$ is not necessarily a constant. Building upon this algorithm, we have achieved a 5-approximation algorithm for the PCMinSSC problem, representing an enhancement over the previously attained approximation ratio of 8 as reported in \cite{liang2023}.

\subsection{Related Work}

Bienstock {\it et al.} \cite{bienstock1993} provided the first constant-approximation algorithm for PCST, achieving a ratio of at most 3. Goemans and Williamson introduced the GW-algorithm in \cite{Goeman1995}, attaining a 2-approximation. Feofiloff {\it et al.} \cite{feofiloff2007} later improved it slightly to $2-2/n$, considering the number of vertices $n$. Notably, the natural linear program relaxation of PCST has an integrality gap of 2. Archer {\it et al.} \cite{archer2011} broke the 2-approximation barrier, presenting an algorithm with a ratio less than 1.9672. The above works are all within the context of general graphs. For Euclidean or planar underlying graphs, polynomial-time approximation schemes (PTASs) were provided in \cite{Chan2020} and \cite{Bateni2011}.


A generalization of PCST is the Prize-Collecting Steiner Forest (PCSF) problem. In PCSF, we are given an edge-weighted graph $G=(V,E,w)$ and a set $\mathcal P=\{(s_i,t_i):s_i,t_i\in V,s_i\neq t_i\}$ of demand pairs, each associated with a nonnegative penalty. The objective in PCSF is to find a subgraph $F\subseteq G$ that minimizes both the edge cost of $F$ and the penalties incurred for unconnected pairs. In \cite{hajiaghayi2006}, Hajiaghayi and Jain demonstrated that a prima-dual approach can also achieve a 3-approximation, and they further provided a $(\frac{1}{1-e^{-1/2}})$-approximation using randomized LP rounding.
PCSF has been extended in \cite{han2017} to cases with demands on sets of vertices instead of pairs, and in \cite{jia2022} to situations where the penalty function is submodular. Both extensions have constant approximation ratios. lt is important to note that PCSF differs from our PCF$_K$ problem, as PCSF requires the forest to connect predefined demand pairs, while PCF$_K$ focuses on the number of forest components without specifying specific vertices or pairs.

Some researchers considered the {\em $k$-prize-collecting Steiner tree} ($k$-PCST) problem, which is a combination of the PCST problem and the {\em $k$-minimum spanning tree} ($k$-MST) problem. A $k$-PCST instance asks for a tree spanning at least $k$ vertices such that the {\em tree-plus-penalty value} is minimized. The problem was first proposed by Han {\it et al.} in \cite{han2019} and a $5$-approximation was obtained using a primal-dual method combined with a Lagrange relaxation method. Later, Matsuda and Takahashi \cite{matsuda2019} obtained approximation ratio $(\alpha+\beta)$, where $\alpha$ and $\beta$ are the approximation ratios of PCST and $k$-MST, respectively. Using the best known ratio 1.9672 for PCST \cite{archer2011} and the best known ratio $2$ for $k$-MST \cite{garg2005}, their ratio for $k$-PCST is at most 3.9672. Later, the approximation ratio was improved to 2 by Pedrosa and Rosado \cite{pedrosa2022} using a modification of the GW-algorithm.


When the cost function is associated with the vertex set rather than the edge set, it gives rise to the Node-Weighted Prize-Collecting Steiner Tree (NW-PCST) problem. Moss {\it et al.} \cite{moss2007} and K\"{o}nemann {\it et al.} \cite{konemann2013}, independently, introduced $O(\ln n)$-approximation algorithms for NW-PCST, which are considered asymptotically optimal unless $P=NP$. In planar graphs, better approximations are attainable: Byrka {\it et al.} \cite{Byrka2016} developed two approximation algorithms for NW-PCST in planar graphs, achieving approximation ratios of 3 and $(2.88+\varepsilon)$, respectively. Additionally, there exists the Node-Weighted Prize-Collecting Steiner Forest (NW-PCSF) problem, which is a node-weighted version of PCSF introduced earlier. In planar graphs, NW-PCSF admits a 4-approximation.

Another generalization of PCST is the {\em group prize-collecting Steiner tree} (GPCST) problem, in which the vertices are partitioned into several groups $\mathcal V=\{V_1,\ldots,V_t\}$, each group is associated with a penalty, a group is said to be  {\em spanned} by a tree $T$ if at least one vertex in the group is contained in $T$, the goal of GPCST is to find a tree $T$ to minimize its edge cost plus the penalty of those groups not spanned by $T$. Glicksman {\it et al.} \cite{glicksman2008} gave a $(2I)$-approximation algorithm, where $I=\max\{|V_1|,\ldots,|V_t|\}$. In \cite{zhang2022}, Zhang {\it et al.} studied the problem assuming that the penalty function is submodular, and also obtained a $(2I)$-approximation.

As previously mentioned, PCF$_K$ was initially introduced as a foundational concept for a sweep cover problem. Sweep cover problems have a wide range of applications, including police patrolling, data gathering, device control, and more \cite{Gao2018,Liang2019,Feng-2015,zhao-2012}. Beginning with the theoretical exploration by Cheng {\it et al.} \cite{Li}, various sweep cover problems have been extensively investigated in the literature.
According to different optimization objectives, current studies on sweep cover can be roughly divided into five categories:

$(\romannumeral1)$ determining the minimum number of mobile sensors \cite{Gorain2015,Chenqq2018,Liang2019,Liu-2021,NIE202170,liang2023};

$(\romannumeral2)$ minimizing the maximum {\em sweep period} \cite{Gao2020,Gao2018};

$(\romannumeral3)$ designing shortest scheduling routes for mobile sensors\cite{Feng-2015,Chen-2021,Liu20201};

$(\romannumeral4)$ determining the minimum sensor speed \cite{zhao-2012,GU-2006};

$(\romannumeral5)$ finding a sweep-cover without requiring all PoIs to be visited, including the PCMinSSC problem \cite{Liang-PCMSSC-TCS} and the {\em budgeted sweep cover} (BSC) problem \cite{Diyan2021,huang2018} the goal of which is to maximize the number of PoIs sweep-covered by a budgeted number of mobile sensors, the {\em min-sensor partial sweep cover} (MinSPSC) problem \cite{liang2023} the goal of which is to minimize the number of mobile sensors to sweep-cover at least $K$ PoIs.

In \cite{Liang-PCMSSC-TCS}, a 5-approximation algorithm (in fact, a 5-LMP algorithm) was obtained for PCMinSSC under the assumption that there is a constant number of predefined base stations and every mobile sensor has to go through some base station. The algorithm is based on a 2-LMP algorithm for RPCF$_K$. While a constant-approximation algorithm for the PCMinSSC problem without base station assumptions have been established in \cite{liang2023}, enhancing the approximation ratio necessitates the development of an approximation algorithm for the unrooted PCF$_K$ problem, as previously discussed.


\subsection{Our Contribution}

This paper's contributions are summarized as follows:
\begin{itemize}
	\item We introduce the first approximation algorithm for the unrooted PCF$_K$ problem, URPCF$_K$. Despite extensive research on the PCST problem, the natural generalization of specifying the number of components in a forest has not been explored until now. With the emergence of the PCMinSSC problem, the study of PCF$_K$ finds practical and theoretical applications. lt is important to note that the prize-collecting Steiner forest problems examined in \cite{hajiaghayi2006}, despite similar names, differ from the one addressed in this paper.
	\item Our URPCF$_K$ algorithm achieves an approximation ratio of at most 2 and can be executed in polynomial time, regardless of whether $K$ is a constant or not. It operates as a lagrangian multiplier preserving algorithm with a factor of 2 (2-LMP). Notably, our algorithm avoids exponential time complexity, which would occur if we guessed $K$ roots (resulting in $O(n^K)$ guesses) and utilized the 2-LMP algorithm for RPCF$_K$  from \cite{Liang-PCMSSC-TCS}. Instead, our approach ensures a more favorable runtime of $O(n^2K^2+mn)$, where $m$ and $n$ are the number of edges and vertices, respectively.
	
	\item To address the issue of exponential time complexity, we propose a rootless growth plus rootless pruning method. Our analysis leverages a rooted analog as a bridge, which is only used in the analysis and does not affect the algorithm's time complexity. We also introduce the concept of net worth value from \cite{Johnson2000}. Using this concept, we demonstrate that the solution produced by our unrooted algorithm is no worse than its rooted counterpart. A key technique for connecting the rooted and unrooted versions while avoiding the guessing of $O(n^K)$ roots is a ``double dynamic programming" method, embedding two dynamic programming algorithms within a primary dynamic programming algorithm. These ideas may have broader applications for solving related problems.
	
	\item The 2-LMP algorithm for URPCF$_K$ is employed to derive a 5-LMP algorithm for PCMinSSC, improving upon the previous 8-approximation presented in \cite{liang2023}. Furthermore, it is worth noting that the 2-LMP for URPCF$_K$ can be adapted into a 2-LMP for the prize-collecting traveling multi-salesman problem ($m$-PCTSP).
\end{itemize}

The rest of this paper is organized as follows: The problem is formally defined in Section~\ref{sec001}, together with some preliminary results. A 2-LMP for URPCF$_K$ is presented in Section~\ref{sec002}, based on which a 5-LMP algorithm for PCMinSSC is presented in Section~\ref{sec003}. Section~\ref{sec004} extends the 2-LMP to the $m$-PCTSP problem and concludes the paper.

\section{Problem formulation and preliminaries}\label{sec001}
This section formally defines the problems studied in this paper.

\begin{definition}[unrooted prize-collecting forest with $K$ components (URPCF$_K$)]\label{def1204-1}
	{\rm Given a positive integer $K$ and a graph $G=(V,E)$ with an edge cost function $w: E\mapsto\mathbb{R}^+$ and a vertex penalty function $\pi: V\mapsto\mathbb{R}^+$, the goal of URPCF$_K$ is to find a forest $F$ with $K$ connected components (a connected component  can also be a singleton), such that the {\em cost-plus-penalty} value is minimized, that is,
$$\min\limits_{F\in\mathcal{F}_K}\{w(F)+\pi(V\setminus V(F))\},$$ where $\mathcal{F}_K$ is the collection of forests with $K$ connected components, $w(F)=\sum_{e\in E(F)}w(e)$ and $\pi(V\setminus V(F))=\sum_{v\notin V(F)}\pi(v)$.}
\end{definition}
The rooted analog of URPCF$_K$, denoted as RPCF$_K$, was first proposed in \cite{Liang-PCMSSC-TCS}. Its formal definition is given as follows.

\begin{definition}[rooted prize-collecting forest with $K$ components (RPCF$_K$)]
	{\rm Given graph $G$, functions $w,\pi$, and integer $K$ as in Definition \ref{def1204-1}, furthermore, let $R=\{r_1,\ldots,r_K\}\subseteq V(G)$ be a set of $K$ pre-specified roots, the goal of RPCF$_K$ is to find a forest $F_R$ with $K$ components, each component contains exactly one root, such that the {\em cost-plus-penalty value} $w(F_R)+\pi(V\setminus V(F_R))$ is minimized.}
\end{definition}

For a positive real number $r\geq1$ (resp. $r\leq1$), a polynomial time algorithm for a minimization (resp. maximization) problem is said to be an {\em $r$-approximation algorithm} if for any instance of the problem, the output of the algorithm has value at most (resp. at least) $r$ times that of an optimal solution. The algorithms presented in the paper for PCF$_K$ are stronger than merely being $r$-approximation, they are $r$-LMP algorithms, whose definition is given as follows.

\begin{definition}[Lagrangian multiplier preserving algorithm with factor $r$ ($r$-LMP) for PCF$_K$ (either rooted or unrooted)]
	{\rm An algorithm for PCF$_K$ is said to be an $r$-LMP if for any instance $I$ of PCF$_K$, the algorithm can compute in polynomial time a forest $F$ with $K$ components such that
		$$w(F)+r\cdot\pi(V\setminus V(F))\leq r\cdot opt(I),$$
		where $opt(I)$ is the optimal value of instance $I$.}
\end{definition}

The studies of PCF$_K$ are motivated by the PCMinSSC problem, and our results are applied to PCMinSSC. A sweep cover problem is motivated by inspecting points of interest periodically. Imagine an application involving static sensors distributed throughout a forest, gathering data on factors like temperature and humidity. This data can be used to predict the risk of a forest fire. To ensure timely fire detection, we want to employ a set of mobile sensors to collect information from the static sensors periodically, let us say every three hours. The challenge is to use as few mobile sensors as possible to accomplish this task. This scenario gives rise to a concept known as the ``sweep cover'', which can be defined as follows:


\begin{definition}[sweep cover]
	{\rm Consider a graph $G$ on vertex set $V$ and edge set $E$. Assume there is a group of mobile sensors moving along the edges of $G$ at speed $a$. Let $t$ be a positive real number representing the sweep period. A vertex $v$ is considered {\em $t$-sweep-covered} if it can be visited at least once within every time period $t$ by a mobile sensor passing through it.}
\end{definition}

The problem of minimum sensor sweep cover was introduced in \cite{Li} as follows:

\begin{definition}[minimum sensor sweep cover (MinSSC)]
{\rm Given a graph $G = (V, E)$, a sweep period $t$, and the speed $a$ of mobile sensors, the objective of MinSSC is to design routes for the minimum number of mobile sensors required to sweep-cover all vertices of $G$.}
\end{definition}

Based on the observation that, in numerous applications, it is more cost-effective to forego serving certain clients and instead pay penalties, \cite{Liang-PCMSSC-TCS} introduced a problem known as the ``prize-collecting sweep cover problem'', which is formally defined as follows:

\begin{definition}[prize-collecting min-sensor sweep cover (PCMinSSC)]
	{\rm Given a graph $G=(V,E)$ with a metric edge length function $l: E\mapsto\mathbb{R^+}$ and a vertex penalty function $\pi: V\mapsto\mathbb{R^+}$, the goal is to design routes for a set of mobile sensors such that the {\em sensor-plus-penalty value} is minimized, that is,
	$$\min\limits_{\mathcal S:\mbox{\small  a set of sensors and their routes}}\{c\cdot|\mathcal{S}|+\pi(V\setminus\mathcal{C}(\mathcal{S}))\},$$
	where $c$ denotes the cost of a mobile sensor, $\mathcal C(\mathcal S)$ is the set of vertices sweep-covered by $\mathcal S$ and  $\pi(V\setminus\mathcal{C}(\mathcal{S}))$ is the total penalty incurred by uncovered vertices.}	
\end{definition}

Making use of our algorithm for URPCF$_K$ on the PCMinSSC problem, we can also obtain an $r$-LMP for PCMinSSC.

\begin{definition}[Lagrangian multiplier preserving algorithm with factor $r$ ($r$-LMP) for PCMinSSC]
	{\rm A polynomial time algorithm is said to be an $r$-LMP for PCMinSSC if for any instance $I$ of PCMinSSC, the computed routes of a set of mobile sensors $\mathcal S$ satisfy
		$$c\cdot|\mathcal S|+r\cdot\sum\limits_{v\notin\mathcal C(\mathcal S)}\pi(v)\leq r\cdot opt(I),$$
		where $opt(I)$ is the optimal value of instance $I$.}
\end{definition}

Building upon the insight that a minimum spanning forest consisting of $K$ components can be derived from a minimum spanning tree by removing the most expensive $K-1$ edges, we consider applying a similar strategy to address the URPCF$_K$ problem. Specifically, this involves removing the $K-1$ most costly edges from an approximate solution to the PCST problem.
To illustrate this approach, consider the example in Fig.~\ref{fig1024}, where the value above a vertex denotes the penalty and the value below an edge denotes its weight. Suppose $K = 2$. An optimal PCST solution would be the entire path. However, if we remove the most expensive edge with a weight of $n + 1$, we obtain a solution for the URPCF$_2$ instance, as indicated by the blackened edges in Fig.~\ref{fig1024} $(b)$. The forest-plus-penalty value for this solution is $n + 2$. An optimal URPCF$_2$ solution is represented by the blackened edges in Fig.~\ref{fig1024} $(c)$, resulting in a forest-plus-penalty value of 3. The approximation ratio in this case is $(n + 2) / 3= \Theta(n)$.
The example shows that this seemingly intuitive approach, which works well for the ``spanning'' analogue, falls short of providing an optimal solution when not all vertices need to be spanned. This observation highlights the need to explore new methods and strategies to tackle the problem effectively.
	
	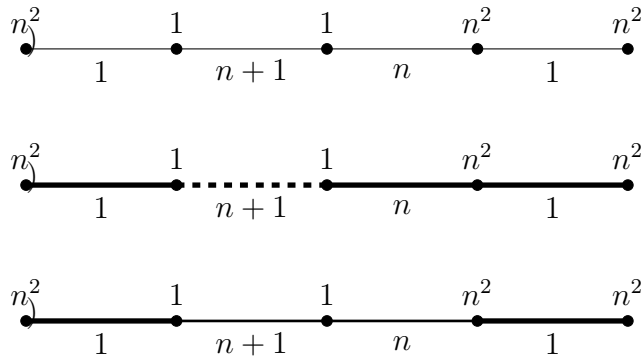
\begin{figure}[H]
		\begin{center}
			\begin{tikzpicture}
				\draw(0,0)--(2,0)--(4,0)--(6,0)--(8,0);
				\filldraw(0,0)circle(2pt)node[above=2pt]{$n^2$};
				\filldraw(2,0)circle(2pt)node[above=2pt]{$1$};
				\filldraw(4,0)circle(2pt)node[above=2pt]{$1$};
				\filldraw(6,0)circle(2pt)node[above=2pt]{$n^2$};
				\filldraw(8,0)circle(2pt)node[above=2pt]{$n^2$};
				\node (P) at (1,-0.3) {$1$};\node (P) at (3,-0.3) {$n+1$};
				\node (P) at (5,-0.3) {$n$};\node (P) at (7,-0.3) {$1$};
\put(250,0){$(a)$}
			\end{tikzpicture}
			\vskip 0.5cm
			\begin{tikzpicture}
				\draw[line width=2pt](0,0)--(2,0);
				\draw[line width=2pt,dashed](2,0)--(4,0);
				\draw[line width=2pt](4,0)--(6,0)--(8,0);
				\filldraw(0,0)circle(2pt)node[above=2pt]{$n^2$};
				\filldraw(2,0)circle(2pt)node[above=2pt]{$1$};
				\filldraw(4,0)circle(2pt)node[above=2pt]{$1$};
				\filldraw(6,0)circle(2pt)node[above=2pt]{$n^2$};
				\filldraw(8,0)circle(2pt)node[above=2pt]{$n^2$};
				\node (P) at (1,-0.3) {$1$};\node (P) at (3,-0.3) {$n+1$};
				\node (P) at (5,-0.3) {$n$};\node (P) at (7,-0.3) {$1$};
\put(250,0){$(b)$}
			\end{tikzpicture}
			\vskip 0.5cm
			\begin{tikzpicture}
				\draw[line width=2pt](0,0)--(2,0);
				\draw[line width=1pt](2,0)--(4,0);
				\draw[line width=1pt](4,0)--(6,0);
				\draw[line width=2pt](6,0)--(8,0);
				\filldraw(0,0)circle(2pt)node[above=2pt]{$n^2$};
				\filldraw(2,0)circle(2pt)node[above=2pt]{$1$};
				\filldraw(4,0)circle(2pt)node[above=2pt]{$1$};
				\filldraw(6,0)circle(2pt)node[above=2pt]{$n^2$};
				\filldraw(8,0)circle(2pt)node[above=2pt]{$n^2$};
				\node (P) at (1,-0.3) {$1$};\node (P) at (3,-0.3) {$n+1$};
				\node (P) at (5,-0.3) {$n$};\node (P) at (7,-0.3) {$1$};
\put(250,0){$(c)$}
			\end{tikzpicture}\caption{An example showing that an idea for URPCF$_K$ which works well when $K=1$ does not work for $K\geq 2$.}\label{fig1024}
		\end{center}
	\end{figure}

\section{A 2-LMP for URPCF$_K$ }\label{sec002}
In this section, we give a 2-LMP algorithm $\mathcal A_1$ for URPCF$_K$. Given an instance $I=(G,K)$ of URPCF$_K$, $\mathcal A_1(I)$ consists of two sub-procedures:
$$
\mbox{Rootless-Growth$(I)$ and Rootless-Prune$(I)$}.
$$
In the following, we shall use $opt_I$ and $OPT_I$ to denote the optimal value and an optimal solution of instance $I=(G,K)$, respectively.

To analyze the performance of $\mathcal A_1$, we make use of an associated rooted PCF$_K$ instance as a step stone. Let $R$ be a set of $K$ vertices of $OPT_I$, taking exactly one vertex from each of the $K$ components of $OPT_I$. Let $I_R=(G,R)$ be an RPCF$_K$ instance. Denote by $opt_{I_R}$ the optimal value of instance $I_R$. Note that
\begin{equation}\label{eq1207-1}
	opt_{I_R}=opt_I.
\end{equation}
It is crucial to emphasize that while the construction of $I_R$ involves utilizing an optimal solution to the problem $I$, this information is employed solely as an intermediate step in our analysis. Importantly, algorithm $\mathcal{A}_1$ does not require knowledge of this information to function effectively.

An algorithm, called $\mathcal A_2$, is designed for the rooted PCF$_k$ problem on instance $I_R$. It consists of three sub-procedures:
$$
\mbox{Rootless-Growth$(I_R)$, $K$-Forest-Step$(I_R)$ and Reverse-Deletion$(I_R)$}.
$$
It is worth highlighting that following the Rootless-Growth step, it is possible for one tree within the forest to have more than one root. The task of the $K$-Forest-Step is to address this situation by splitting some of these trees into smaller ones, ensuring that each tree contains precisely one root. This aspect is a novel addition that distinguishes it from the PCST problem. Additionally, this introduces new considerations and complexities into the Reverse-Deletion step of the algorithm.


Denote by $pc(\mathcal A_i)$ the {\em cost-plus-penalty} value of the forest computed by algorithm $\mathcal A_i$ (takeing $r=2$ in the definition of $r$-LMP into consideration, that is, $pc(\mathcal A_i)=w(F_i)+2 \pi(V\setminus V(F_i))$ where $F_i$ is the output forest of algorithm $\mathcal A_i$). We shall prove
\begin{equation}\label{eq1207-2}
	pc(\mathcal A_1)\leq pc(\mathcal A_2)\leq 2opt_{I_R}.
\end{equation}
By combining this observation with equation \eqref{eq1207-1}, we can derive an approximation ratio of 2 for algorithm $\mathcal A_1$. For ease of understanding, we present the algorithms in the order of $\mathcal A_2,\mathcal A_1$ in the following.

\subsection{$\mathcal A_2$: Rootless-Growth+ $K$-Forest-Step+ Reverse-Deletion} \label{sec3-1}
In this section, we introduce a 2-LMP algorithm for instance $I_R$ of RPCF$_K$, utilizing a primal-dual approach. In a previous work \cite{Liang-PCMSSC-TCS}, a 2-LMP algorithm for RPCF$_K$ was developed, comprising two key steps: rooted-growth and reverse-deletion. However, this existing algorithm cannot be directly applied as a building block for the analysis of our unrooted PCF$_K$ algorithm. Consequently, we need to introduce a different algorithm tailored to URPCF$_K$. The modification involves replacing the rooted-growth method with a rootless-growth approach while maintaining the same approximation ratio. To achieve this, we initiate the process with an LP formulation of RPCF$_K$.

\subsubsection{LP formulation}\label{RLP}
Suppose $G=(V,E)$ is a graph and $R=\{r_1,\ldots,r_K\}\subseteq V$ is a set of roots. For a forest $F$, define two types of indicator variables $\{x(e)\}_{e\in E}$ and $\{z(T)\}_{T\subseteq V-R}$ as follows. For each edge $e\in E$, let
\begin{equation}\label{eq1014-1}
	x(e)=\left\{\begin{array}{ll}
		1, &\mbox{if $e\in F$},\\
		0, &\mbox{otherwise}. \end{array}\right.
\end{equation}
We say that a vertex set $T$ {\em is not spanned by $F$} if no vertex of $T$ is spanned by $F$. For each vertex set $T\subseteq V-R$, let
\begin{equation}\label{eq1014-2}
	z(T)=\left\{\begin{array}{ll}
		1, &\mbox{if $T$ is not spanned by $F$},\\
		0, &\mbox{otherwise}. \end{array}\right.
\end{equation}
An edge $e$ is said to {\em cut} a vertex set $S$ if $e$ has exactly one end in $S$. Denote by $\delta(S)$ the set of edges that cut $S$, and define  $x(\delta(S))=\sum_{e\in \delta(S)}x(e)$. The LP relaxation corresponding to an integer linear program (ILP) of an RPCF$_K$ instance can be written as follows, which is the same as that in \cite{Liang-PCMSSC-TCS}:

\begin{align}\label{eq1208-1}
	\min &\sum_{e\in E} w(e)x(e)+\sum_{T\subseteq V-R}\left(\sum_{v\in T}\pi(T)\right)z(T)\\
	\mbox{s.t.} \ &x(\delta(S))+\sum_{T\subseteq V-R: S
	\subseteq T}z(T)\geq 1, \ \forall S\subseteq V-R,\nonumber\\
&x(e)\geq 0, \ \forall e\in E,\nonumber\\
&z(T)\geq 0,\ \forall T\subseteq V-R.\nonumber
\end{align}
The explantation for the ILP corresponding to \eqref{eq1208-1} is as follows. An optimal integer solution corresponds to a forest $F$, and there is only one vertex set $T$ having $z(T)=1$, namely $T_F=V-V(F)$. It is the biggest vertex set not spanned by $F$. The first constraint says that $\forall S\subseteq V-R$, either it is cut by an edge of $F$, or it is contained in $T_F$.

The dual LP of \eqref{eq1208-1} is:
\begin{align}\label{lp2}
	\max &\sum_{S\subseteq V-R}y_S\\
	\mbox{s.t.}&\ \sum_{S\subseteq V-R: e\in\delta(S)}y_S\leq w(e), \ \forall e\in E, \label{eq1208-2}\\
	&\sum_{S\subseteq T}y_S\leq \sum_{v\in T}\pi(v),\ \forall T\subseteq V-R, \label{eq1208-3}\\
	&y_S\geq 0, \ \forall S\subseteq V-R. \nonumber
\end{align}
In fact, our rootless-growth is based on the following {\em strengthening} of this dual LP:
\begin{align}\label{eq1221-1}
	\max &\sum_{S\subseteq V-R}y_S\\
	\mbox{s.t.}&\ \sum_{S\subseteq V: e\in\delta(S)}y_S\leq w(e), \ \forall e\in E, \label{eq1221-2}\\
	&\sum_{S\subseteq T}y_S\leq \sum_{v\in T}\pi(v),\ \forall T\subseteq V, \label{eq1221-3}\\
	&y_S\geq 0, \ \forall S\subseteq V. \nonumber
\end{align}
Compared with the previous dual LP, the variables in the
relaxed dual LP also include those $y_S$ corresponding to set $S$ containing roots, the left term of the first constraint is changed into the summation over $S\subseteq V$ instead of $S\subseteq V-R$, and the second type of constraints are for any $T\subseteq V$ instead of $T\subseteq V-R$. The modified constraints are tighter, and thus the optimal value of the relaxed dual LP \eqref{eq1221-1} is a lower bound for the optimal value of \eqref{lp2}, and in turn a lower bound for the optimal value of \eqref{eq1208-1} (by duality theory).

\subsubsection{Ideas of the algorithm}
First, let us provide a concise summary of the key concepts from the rooted-growth and reverse-deletion method outlined in \cite{Liang-PCMSSC-TCS}. Then, we will highlight the differences between this approach and the rootless-growth and $K$-forest-step methods introduced in this paper.

In \cite{Liang-PCMSSC-TCS}, the algorithm begins with a zero-solution of the dual LP \eqref{lp2} and incrementally increases the dual variables $y_S$ associated with active components (the meaning of ``active'' will be explained in details in the following) until specific constraints become tight, meaning that inequalities become equalities. When a constraint of the form \eqref{eq1208-2} becomes tight, the corresponding edge is added to the forest $F$. If a constraint of the form \eqref{eq1208-3} becomes tight, the vertex set $T$ related to this constraint is deactivated, implying that the dual variable $y_T$ is frozen and will not be further increased. This process continues until all vertex sets are inactive. To achieve this in polynomial time, the algorithm only considers dual variables associated with connected components of $G_F = (V,F)$, which are termed effective dual variables. Other dual variables, which number exponentially, are implicitly set to zero. Initially, $G_{\emptyset}=(V,\emptyset)$ comprises $|V|$ {\em trivial} connected components. Every vertex in $R$ is set to be an inactive component, while those vertices in $V-R$ are set to be active components. Only active components can have their dual variables increased. During the process, a component containing a root always remains inactive. Therefore, when two components merge due to the addition of an edge between them, and one of them is inactive, the merged new component also becomes inactive. Notably, no two components containing roots can merge because they are both inactive, thus unable to have their dual variables increased. An edge $e$ can only change from loose to tight when at least one of the two components it cuts has an increased dual variable. Consequently, each component contains at most one root. At the end of the rooted-growth step, components with exactly one root are retained. In the reverse-deletion step, edges cutting inactive leaf components are deleted.

The rootless-growth step in this paper differs from the rooted-growth step in two main ways. First, it is based on dual LP \eqref{eq1221-1} instead of dual LP \eqref{lp2}. Second, all components, regardless of whether they contain roots, are treated equally. A component becomes inactive only when the corresponding constraint for its vertex set becomes tight. However, this approach can lead to some components containing more than one root by the end of the rootless-growth step. This is where the $K$-forest step comes in: it involves deleting certain costly edges to ensure that each component contains at most one root. As before, only components with exactly one root are retained, and reverse-deletion is applied to the forest obtained after the $K$-forest-step.

\subsubsection{The algorithm and some explanation}
The details of the algorithm are given in Algorithm~\ref{algo1}.

\begin{algorithm}[htbp]
	\caption {Algorithm for RPCF$_K$}
	\hspace*{0.02in}\raggedright{\bf Input:} A graph $G=(V,E)$ with edge cost function $w$ and vertex penalty function $\pi$, a positive integer $K$ and a set $R=\{r_1,\ldots,r_K\}$ of $K$ roots.
	
	\hspace*{0.02in}\raggedright{\bf Output:} A forest with $K$ connected components, each component contains one root.
	
	\begin{algorithmic}[1]
		\STATE \textbf{// Rootless-Growth}
		\STATE $F\leftarrow\emptyset$; $\mathcal{C}=\{\{v\}:v\in V\}$; $\forall v\in V$, $\kappa(\{v\})\leftarrow 1$, $d(v) \leftarrow 0$, $y_{\{v\}}\leftarrow 0$, $h(\{v\})\leftarrow 0$;
		\WHILE{$\exists$ an active component in $\mathcal{C}$}
			\STATE $\varepsilon_1\leftarrow \min_{e=uv, u\in S_u\in\mathcal{C}, v\in S_v\in\mathcal{C}, S_u\neq S_v}\left\{ \frac{w(e)-d(u)-d(v)}{\kappa(S_u)+\kappa(S_v)}\right\}$;
			$e^*=u^*v^*\leftarrow \arg\varepsilon_1$;\label{line1222-1}
			\STATE $\varepsilon_2\leftarrow\min_{S\in\mathcal{C}:\kappa(S)=1}\left\{\sum_{v\in S}\pi(v)-h(S)\right\}$;
			$S^*\leftarrow \arg\varepsilon_2$;
			\STATE $\varepsilon\leftarrow\min\{\varepsilon_1,\varepsilon_2\}$; \label{line1222-2}
			\FORALL{$S\in \mathcal{C}$ with $\kappa(S)=1$}
			\STATE $y_S\leftarrow y_S+\varepsilon$, $h(S)\leftarrow h(S)+\varepsilon$;  $\forall v\in S$, $d(v)\leftarrow d(v)+\varepsilon$;
			\ENDFOR
				\IF{$\varepsilon_2<\varepsilon_1$}
					\STATE $\kappa(S^*)\leftarrow 0$
				\ELSE
					\STATE $F\leftarrow F+e^*$, $S\leftarrow S_{u^*}+S_{v^*}+e^*$, $\kappa(S)\leftarrow 1$, $\mathcal{C}\leftarrow (\mathcal{C}-\{S_{u^*},S_{v^*}\})\cup \{S\}$;
				\ENDIF
		\ENDWHILE
		\STATE \textbf{//$K$-Forest-Step}
		\STATE $F_R\leftarrow$ the set of  components of $\mathcal C$ containing roots;\label{line1222-5}
			\WHILE{some component $S_R\in F_R$ contains more than one root }
        \STATE $r_i,r_j\leftarrow$ two nearest roots on $S_R$;
		\STATE	$P_{ij}\leftarrow$ the unique path in $S_R$ connecting $r_i,r_j$;
		\STATE $e_{ij}\leftarrow$ the last edge added into $F$ among all edges of $P_{ij}$;
		\STATE $F_R\leftarrow F_R-e_{ij}$
		\ENDWHILE
	
		\STATE \textbf{// Reverse-Deletion}
		\FOR{each $T_r\in F_R$ with $|V(T_r)|\geq 2$}
		\FOR{edges $e$ in $T_r$ in reverse order of their additions into $F$}
			\IF{$e$ cuts an inactive leaf component $C_e$ of $T_r$ when $e$ is added}
				\STATE $T_r\leftarrow T_r-e-C_e$
			\ENDIF
		\ENDFOR
		\ENDFOR
		\RETURN $F_R\leftarrow\{T_r\}_{r\in R}$
	\end{algorithmic}\label{algo1}
\end{algorithm}

In the algorithm, $F$ is the set of added edges and $\mathcal C$ is the set of connected components of $(V(G),F)$ (the spanning subgraph of $G$ induced by $F$). For the ease of statement, we will use $S$ to refer to both a component of $\mathcal C$ and its vertex set. A dual variable $y_S$ is called an {\em effective dual variable} if $S$ is a component of $\mathcal C$. We use $\kappa(S)$ to indicate whether $S\in\mathcal C$ is an {\em active component} or not, i.e.,
	\begin{equation*}
		\kappa(S)=\left\{\begin{array}{ll}
			1, &\mbox{if $S$ is an active componnt},\\
			0, &\mbox{if $S$ is an inactive componnt}. \end{array}\right.
	\end{equation*}
A dual variable $y_S$ corresponding to an active component $S\in\mathcal C$ is called an {\em active effective dual variable}. Only active effective dual variables can have their values increased.

Initially, $F=\emptyset$, every vertex $v\in V$ is a trivial active component of $\mathcal{C}$, and all dual variables are implicitly set to be zeros. In each iteration, all active effective dual variables increase their values simultaneously until some constraint of the dual LP \eqref{eq1221-1} becomes tight. A constraint with the form \eqref{eq1221-2} is called an {\em edge-constraint}.
And for any $T\subseteq V$ (including the vertex set containing some root), a constraint with the form \eqref{eq1221-3} is called a {\em set-constraint}. For simplicity of statement, we also say an edge becomes tight or a vertex set becomes tight. If an edge $e$ becomes tight, then $e$ is added into $F$, and the two components containing the two ends of $e$ are merged into one component of $\mathcal C$. If a vertex set $T$ becomes tight, then $T$ is {\em deactivated} and its dual variable $y_T$ is {\em frozen} (i.e., $y_T$ will never increase anymore).

In order to detect the first constraint that becomes tight in polynomial time, we use $d(v)$ to record the accumulated amount of increase on vertex $v$, i.e.,
$$d(v)=\sum_{S:v\in e\in\delta(S)}y_S,$$
where $v\in e$ means $v$ is incident with edge $e$. For a vertex set $T$, we use $h(T)$ to record the accumulated amount of increase on those subsets of $T$, i.e.,
$$h(T)=\sum_{S:S\subseteq T}y_S.$$
With the help of these notations, an edge-constraint for edge $e=uv$ can be written as
\begin{equation}\label{eq1221-4}
d(u)+d(v)\leq w(e),
\end{equation}
and a set-constraint for vertex set $T$ can be written as
\begin{equation}\label{eq1221-5}
h(T)\leq \pi(T)
\end{equation}
Suppose in an iteration, the increase of each active effective dual variable is $\varepsilon$. If it is edge $e=uv$ that first becomes tight, let $S_u$ and $S_v$ be the two components of $\mathcal C$ at the beginning of this iteration, that contain $u$ and $v$, respectively, then
$$
d(u)+d(v)+\varepsilon (\kappa(S_u)+\kappa(S_v))=w(e).
$$
Hence,
$$
\varepsilon=\frac{w(e)-d(u)-d(v)}{\kappa(S_u)+\kappa(S_v)}.
$$
If it is vertex set $T$ that first becomes tight, let $\mathcal A$ (resp. $\mathcal I$) be the set of active (resp. inactive) components of $\mathcal C$ at the beginning of this iteration that are contained in $T$. Note that $\mathcal A\neq\emptyset$ since otherwise $h(T)$ will not increase. For any $S_I\in\mathcal{I}$, it is deactivated in a previous iteration, and thus
\begin{equation}\label{eq1222-2}
h(S_I)= \pi(S_I).
\end{equation}
In this iteration, every $S_A\in\mathcal{A}$ has its dual variables $y_{S_A}$ increased by $\varepsilon$. So,
\begin{equation}\label{eq1222-3}
h(T)+\varepsilon\cdot|\mathcal{A}|=\pi(T).
\end{equation}
Combining \eqref{eq1222-2} and \eqref{eq1222-3} with the observation
\begin{align*}
& h(T)=\sum_{S_A:S_A\in\mathcal{A}}h(S_A)+\sum_{S_I:S_I\in\mathcal{I}}h(S_I) \ \mbox{and}\\
& \pi(T)=\sum_{S_A:S_A\in\mathcal{A}}\pi(S_A)+\sum_{S_I:S_I\in\mathcal{I}}\pi(S_I),
\end{align*}
we have
\begin{equation}\label{eq1222-11}
\sum_{S_A:S_A\in\mathcal{A}}h(S_A)+\varepsilon\cdot|\mathcal{A}|=\sum_{S_A:S_A\in\mathcal{A}}(h(S_A)+\varepsilon)=\sum_{S_A:S_A\in\mathcal{A}}\pi(S_A).
\end{equation}
Since the algorithm keeps a {\em feasible} dual solution to LP \eqref{eq1221-1}, any $S_A\in\mathcal A$ satisfies
$$
h(S_A)+\varepsilon\leq \pi(S_A).
$$
Hence the second equality in \eqref{eq1222-11} implies
$$
h(S_A)+\varepsilon=\pi(S_A) \ \mbox{for any $S_A\in\mathcal A$},
$$
that is, every active component of $\mathcal C$ that is contained in $T$ becomes tight at the same time as $T$. So, to find out the first vertex set that becomes tight, it is sufficient to consider those active components, and the increased amount can be determined by
$$
\varepsilon=\min_{S:\mbox{\scriptsize active component of} \ \mathcal C}\{\pi(S)-h(S)\}.
$$
The above arguments are the theoretical basis for line \ref{line1222-1} to line \ref{line1222-2} of Algorithm \ref{algo1}.

The rootless-growth step stops when all components in $\mathcal{C}$ are inactive. Let $F_R$ be the set of components of $\mathcal C$ containing roots. If a component of $F_R$ contains more than one roots, then the $K$-forest step breaks it into smaller components. This is realized by the following method. Suppose a component $S_R\in F_R$ contains more than one roots; since $S_R$ is a tree, any pair of roots $r_i,r_j\in V(S_R)$ are linked by a unique path in $S_R$, denote this path as $P_{ij}$; choose roots $r_i,r_j\in V(S_R)$ such that there is no other roots on $P_{ij}$; among all edges of $P_{ij}$, delete the edge $e_{ij}$ which is the last one added into $F$; replace $S_R$ by the two components of $S_R-e_{ij}$ in the new forest $F_R$; proceed in this way until each component of $F_R$ contains exactly one root. The algorithm uses $E_K$ to denote the set of removed edges. As to why the ``last'' edge is removed for breaking, we shall add a remark at the end of Section \ref{sec1029-1}.
%

The final output is obtained from $F_R$ by a {\em reverse deletion} step described as follows. For each component $T_r$ of $F_R$, consider its edges in reverse order of their additions into $F$ and  decide whether an edge $e$ and some component incident with $e$ can be deleted or not from $T_r$ according to the following rule: let $\mathcal C_e$ be the collection of components just before $e$ is added into $F$; adding $e$ will merge two components of $\mathcal C_e$; if one of them

$(a)$ is inactive,

$(b)$ does not contain root, and

$(c)$ is not incident with any other edge in the final forest $F_R$,

\noindent then remove $e$ and this component (call it a {\em  rootless inactive leaf component}) from $T_r$.

An example of a constructed subtree $T_r$ is shown in Fig.~\ref{fig1230-1}, in which a tight edge is marked as $e_i$ and a deactivated component is marked as $S_i$, and the indices indicate the temporal order of events, i.e., $\{e_1,e_2,S_3,e_4,S_5,e_6,S_7,e_8\}$.

\begin{figure}[htpb]
	\begin{center}
		\begin{tikzpicture}
			\filldraw(0,0)circle(2pt);
			\filldraw(0,2)circle(2pt);
			\filldraw(2,0)circle(2pt);
			\filldraw(2,2)circle(2pt);
			\filldraw(4,0)circle(2pt);
			\filldraw(4,2)circle(2pt);
\draw[line width = 1pt](2,0)--(2,2);
\draw[line width = 1.5pt](2,1)ellipse(0.5 and 1.5);
\node (r1) at (2.2,1) {$e_1$};
\draw[line width = 1pt](4,0)--(4,2);
\draw[line width = 0.5pt](4,1)ellipse(0.5 and 1.5);
\node (r1) at (4.2,1) {$e_2$};\node (r1) at (4.6,1.7) {$S_3$};
\draw[line width = 1pt](2,2)--(4,2);\node (r1) at (3,2.2) {$e_4$};
\draw[line width = 1.5pt](3,1)circle(2);
\draw[line width = 0.5pt](0,2)circle(8pt);\node (r1) at (0.2,2.2) {$S_5$};
\draw[line width = 1pt](0,0)--(0,2);\node (r1) at (-0.2,1) {$e_6$};
\node (r1) at (-0.2,0) {$r$};
\draw[line width = 0.5pt](0,1)ellipse(0.6 and 1.8);\node (r1) at (0.5,-0.3) {$S_7$};
\draw[line width = 1pt](2,0)--(0,2);\node (r1) at (0.8,1.4) {$e_8$};
		\end{tikzpicture}
	\end{center}
\caption{A constructed subtree $T_r$. Blackened circles are active components, light circles indicate deactivated components.}\label{fig1230-1}
\end{figure}
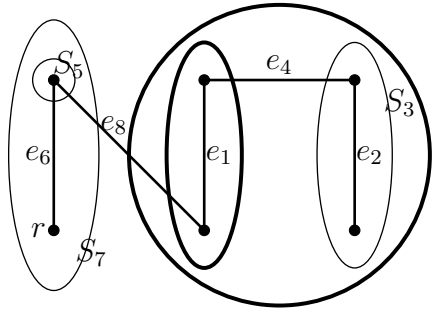

The reverse-deletion on $T_r$ is given in Fig.~\ref{fig4}. According to ``reverse'' deletion rule, the first edge considered is $e_8$. Although it cuts an inactive leaf component, the component contains a root, so $e_8$ is kept. Then consider edge $e_6$. The leaf component cut by $e_6$ is active, so $e_6$ is kept. For $e_4$, it cuts a rootless inactive leaf component, namely $S_3$, so both $e_4$ and $S_3$ are removed. Finally, for edge $e_1$, both components incident with $e_1$ are active, so $e_1$ is intact.

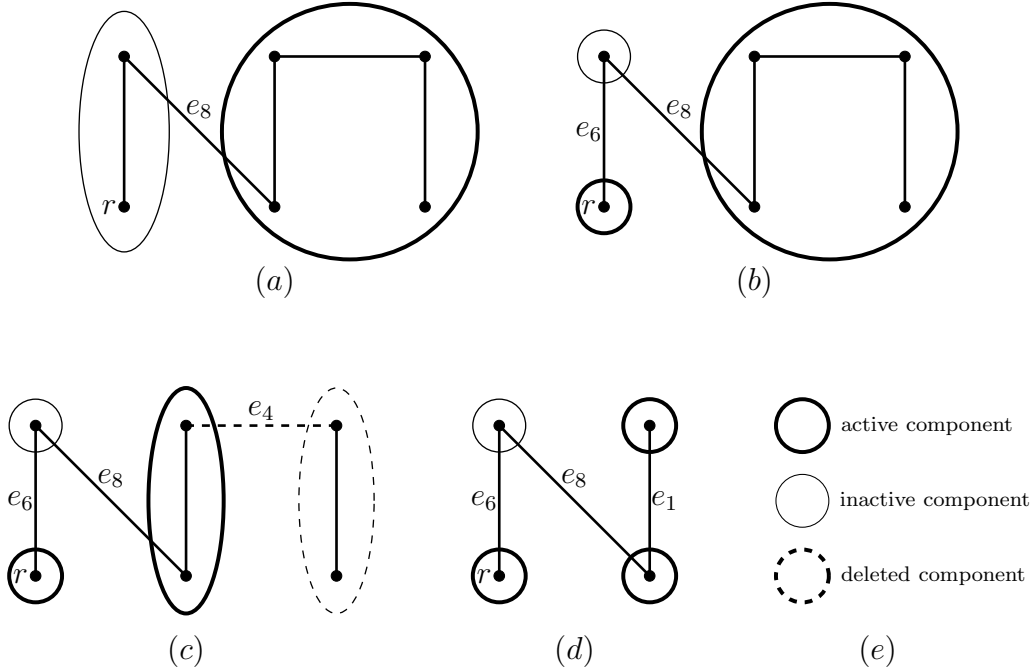
\begin{figure}[htpb]
			\begin{center}
		\begin{tikzpicture}
			\filldraw(0,0)circle(2pt);
\filldraw(0,2)circle(2pt);
\filldraw(2,0)circle(2pt);
\filldraw(2,2)circle(2pt);
\filldraw(4,0)circle(2pt);
\filldraw(4,2)circle(2pt);
\draw[line width = 1pt](2,0)--(2,2);
\draw[line width = 1pt](4,0)--(4,2);
\draw[line width = 1pt](2,2)--(4,2);
\draw[line width = 1.5pt](3,1)circle(1.7);
\draw[line width = 1pt](0,0)--(0,2);
\node (r1) at (-0.2,0) {$r$};
\draw[line width = 0.5pt](0,1)ellipse(0.6 and 1.6);
\draw[line width = 1pt](2,0)--(0,2);\node (r1) at (1,1.3) {$e_8$};
\node (r1) at (2,-1) {$(a)$};
		\end{tikzpicture}
		\hskip 1cm
		\begin{tikzpicture}
			\filldraw(0,0)circle(2pt);
\filldraw(0,2)circle(2pt);
\filldraw(2,0)circle(2pt);
\filldraw(2,2)circle(2pt);
\filldraw(4,0)circle(2pt);
\filldraw(4,2)circle(2pt);
\draw[line width = 1pt](2,0)--(2,2);
\draw[line width = 1pt](4,0)--(4,2);
\draw[line width = 1pt](2,2)--(4,2);
\draw[line width = 1.5pt](3,1)circle(1.7);
\draw[line width = 1pt](0,0)--(0,2);
\node (r1) at (-0.2,0) {$r$};
\draw[line width = 1pt](2,0)--(0,2);\node (r1) at (1,1.3) {$e_8$};
\node (r1) at (-0.2,1) {$e_6$};
\draw[line width = 1.5pt](0,0)circle(10pt);
\draw[line width = 0.5pt](0,2)circle(10pt);
\node (r1) at (2,-1) {$(b)$};
		\end{tikzpicture}
		\vskip 1cm
		\begin{tikzpicture}
			\filldraw(0,0)circle(2pt);
\filldraw(0,2)circle(2pt);
\filldraw(2,0)circle(2pt);
\filldraw(2,2)circle(2pt);
\filldraw(4,0)circle(2pt);
\filldraw(4,2)circle(2pt);
\draw[line width = 1pt](2,0)--(2,2);
\draw[line width = 1pt](4,0)--(4,2);
\draw[line width = 1pt,dashed](2,2)--(4,2);
\draw[line width = 1pt](0,0)--(0,2);
\node (r1) at (-0.2,0) {$r$};
\draw[line width = 1pt](2,0)--(0,2);\node (r1) at (1,1.3) {$e_8$};
\node (r1) at (-0.2,1) {$e_6$};
\draw[line width = 1.5pt](0,0)circle(10pt);
\draw[line width = 0.5pt](0,2)circle(10pt);
\node (r1) at (3,2.2) {$e_4$};
\draw[line width = 1.5pt](2,1)ellipse(0.5 and 1.5);
\draw[line width = 0.5pt,dashed](4,1)ellipse(0.5 and 1.5);
\node (r1) at (2,-1) {$(c)$};
		\end{tikzpicture}
		\hskip 1cm
		\begin{tikzpicture}
			\filldraw(0,0)circle(2pt);
\filldraw(0,2)circle(2pt);
\filldraw(2,0)circle(2pt);
\filldraw(2,2)circle(2pt);
\draw[line width = 1pt](2,0)--(2,2);
\draw[line width = 1pt](0,0)--(0,2);
\node (r1) at (-0.2,0) {$r$};
\draw[line width = 1pt](2,0)--(0,2);\node (r1) at (1,1.3) {$e_8$};
\node (r1) at (-0.2,1) {$e_6$};
\draw[line width = 1.5pt](0,0)circle(10pt);
\draw[line width = 0.5pt](0,2)circle(10pt);
\draw[line width = 1.5pt](2,0)circle(10pt);
\draw[line width = 1.5pt](2,2)circle(10pt);
\node (r1) at (2.2,1) {$e_1$};
\node (r1) at (1,-1) {$(d)$};
		\end{tikzpicture}
		\hskip 1cm
		\begin{tikzpicture}
			\draw(0,1)[line width=1.5pt]circle(10pt)node[right=10pt]{\scriptsize active component};
			\draw(0,0)circle(10pt)node[right=10pt]{\scriptsize inactive component};
			\draw[line width=1.5pt,dashed](0,-1)circle(10pt)node[right=10pt]{\scriptsize deleted component};
			\node (r1) at (1,-2) {$(e)$};

		\end{tikzpicture}
	\end{center}
	\caption{An illustration of reverse-deletion.}\label{fig4}
\end{figure}

\subsubsection{Performance analysis}\label{sec1029-1}

\begin{theorem}\label{theo1}
{\rm	Algorithm~\ref{algo1} is a 2-LMP for RPCF$_K$ and runs in time $O(nm)$, where $n$ is the number of vertices and $m$ is the number of edges.}
\end{theorem}

\begin{proof}
Note that throughout the algorithm, $\{y_S\}_{S\subseteq V}$ is a feasible solution to the relaxed dual LP \eqref{eq1221-1}. By the remark at the end of Section \ref{RLP}, we have
\begin{equation}\label{eq1228-2}
\sum_{S\subseteq V-R}y_S\leq opt,
\end{equation}
where $opt$ is the optimal value for the RPCF$_K$ instance.
	Let $X$ be the set of vertices not spanned by the final output forest $F_R$. To prove the theorem, it suffices to prove
	\begin{align}\label{eq108-1}
		\sum_{e\in F_R}w(e)+2\sum_{v\in X}\pi(v)\leq 2\sum_{S\subseteq V-R}y_S.
	\end{align}
Observe that the set $X$ is the disjoint union of some inactive components, because only inactive components can be deleted in the reverse deletion step. Suppose \begin{equation}\label{eq1224-11}
X=V(C_1)\cup\cdots\cup V(C_q),
\end{equation}
where $C_1,\ldots,C_q$ are inactive components. Since a vertex set is deactivated only when its corresponding constraint becomes tight, we have
\begin{align}\label{eq108-2}
	\sum_{v\in X}\pi(v)=\sum_{j=1}^{q}\sum_{v\in V(C_j)}\pi(v)=\sum_{j=1}^{q}\sum_{S\subseteq V(C_j)}y_S.
\end{align}
Since edge $e$ is added only when the corresponding edge constraint becomes tight, we have
\begin{align}\label{eq108-3}
	\sum_{e\in F_R}w(e)=\sum_{e\in F_R}\sum_{S:e\in\delta(S)}y_S=\sum_{S\subseteq V}y_S\cdot|F_R\cap \delta(S)|.
\end{align}
Therefore, to prove \eqref{eq108-1}, it suffices to prove
\begin{align}\label{eq108-4}
	\sum_{S\subseteq V}y_S\cdot|F_R\cap \delta(S)|+2\sum_{j=1}^{q}\sum_{S\subseteq V(C_j)}y_S\leq 2\sum_{S\subseteq V-R}y_S.
\end{align}

Denote by $\{y_S^{(i)}\}_{S\subseteq V}$ the set of dual variables at the end of the $i$-th iteration of the first while-loop. We prove inequality \eqref{eq108-4} (replace every $y_S$ by $y_S^{(i)}$) by induction on $i$. Inequality \eqref{eq108-4} trivially holds for $i=0$ because $\{y_v^{(0)}\}_{v\in V}=\{0\}_{v\in V}$. Assume that \eqref{eq108-4} holds for $(i-1)$, that is,
\begin{align}\label{eq108-5}
	\sum_{S\subseteq V}y_S^{(i-1)}\cdot|F_R\cap \delta(S)|+2\sum_{j=1}^{q}\sum_{S\subseteq V(C_j)}y_S^{(i-1)}\leq 2\sum_{S\subseteq V-R}y_S^{(i-1)}.
\end{align}
We shall prove that during the $i$-th iteration,
\begin{equation}\label{eq1224-13}
\begin{array}{c}
\mbox{the increase to the left-hand side of inequality \eqref{eq108-5}}\\ \mbox{is no more than the increase in the right-hand side.}
\end{array}
\end{equation}

For this purpose, construct a {\em condensation graph} $H$ as follows: consider the collection of components $\mathcal C^{(i)}$ at the beginning of the $i$-th iteration; contract every component in $\mathcal C^{(i)}$ into a super node; keep all those edges $e$ with the form $e\in(\delta (S)\cap F_R)$ for some component $S\in\mathcal C^{(i)}$, where $F_R$ is the final forest output by the algorithm.
For clarity of statement, we call vertices of $H$ as nodes.
Note that some component $S\in\mathcal C^{(i)}$ may contain more than one roots. See Fig. \ref{fig1224-1} for an illustration. There are five components in $\mathcal C^{(i)}$. The blackened edges belong to the final forest $F_R$. The condensation graph is depicted in Fig. \ref{fig1224-1} $(b)$. Note that $S_1$ contains two roots, and thus the $K$-forest step has to remove some edge from the path linking them (in this figure, the dashed edge $\hat e$ is removed).

\begin{figure}[htpb]
\begin{center}
\begin{picture}(130,90)
\put(20,30){\oval(20,30)}\put(20,22){\circle*{5}}\put(20,38){\circle*{5}}\put(20,22){\line(0,1){16}}
\put(50,15){\oval(30,20)}\put(42,15){\circle*{5}}\put(58,15){\circle*{5}}\put(42,15){\linethickness{0.5mm}\line(1,0){16}}
\put(80,15){\oval(20,20)}\put(80,16){\circle*{5}}{\linethickness{0.5mm}\qbezier(80,16)(92,19)(104,22)}
\put(104,30){\oval(20,30)}\put(104,22){\circle*{5}}\put(104,38){\circle*{5}}\put(104,22){\line(0,1){16}}
\put(65,65){\oval(100,35)}\put(30,54){\circle*{5}}\put(30,74){\circle*{5}}\put(50,64){\circle*{5}}\put(80,64){\circle*{5}}\put(100,74){\circle*{5}}
\put(30,54){\line(2,1){20}}\multiput(54,64)(4,0){6}{\line(1,0){2.5}}
{\linethickness{0.5mm}\qbezier(30,74)(40,69)(50,64)\qbezier(80,64)(90,69)(100,74)}
\put(5,55){$S_1$}\put(-1,18){$S_2$}\put(113,18){$S_3$}\put(27,2){$S_4$}\put(87,2){$S_5$}
\put(46,55){$r_1$}\put(74,55){$r_2$}\put(62,66){$\widehat{e}$}\put(40,18){$r_3$}\put(74,8){$r_4$}
\put(60,-10){(a)}
\end{picture}
\hskip 1cm\begin{picture}(70,90)
\put(40,60){\circle*{5}}\put(70,40){\circle*{5}}\put(50,20){\circle*{5}}
\put(20,20){\circle*{5}}\put(0,40){\circle*{5}}
\put(70,40){\line(-1,-1){20}}
\put(44,58){$S_1$}\put(4,37){$S_2$}\put(73,38){$S_3$}\put(21,12){$S_4$}\put(51,12){$S_5$}
\put(32,-10){(b)}
\end{picture}
\vskip 0.2cm
\caption{An illustration of the condensation graph. Circles indicate components of $\mathcal C^{(i)}$. Roots are labeled by $r$. Blackened edges belong to the final forest $F_R$. The dashed edge $\hat e$ is removed in the $K$-forest step. }\label{fig1224-1}
\end{center}
\end{figure}
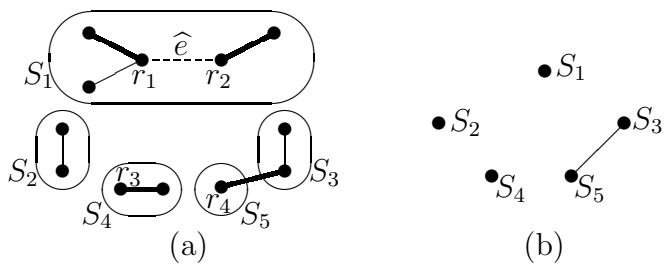

For a node $v\in V(H)$, denote by $C_v$ the component of $\mathcal C^{(i)}$ corresponding to $v$. We call $v$ a {\em root-node} if $C_v$ contains a root. The following claim plays an important role in analyzing the approximation ratio.

\vskip 0.2cm
{\bf Claim 1.} Every connected component of graph $H$ has at most one root-node.

Suppose $H'$ is a connected component of $H$ having more than one root-nodes. Blossom $H'$, that is, replace every node $v\in V(H')$ by its corresponding component $C_v$. Note that the blossomed graph, denoted as $G'$, is still a tree, whose edges come from $F_R^b$, the forest at the end of the rootless-growth step, before deleting edges. Consider two nearest roots $r_j,r_k$ in $G'$ belonging to different root-nodes $v_j$ and $v_k$. There is a unique path $P_{jk}$ in $G'$ linking $r_j$ and $r_k$. By the $K$-forest step, an edge $e_{jk}$ is removed. Remember that $e_{jk}$ is the {\em last} edge on $P_{jk}$ added into $F$. Also note that
\begin{equation}\label{eq1227-1}
\begin{array}{c}
\mbox{all those edges within components of $\mathcal C^{(i)}$ are added}\\
\mbox{{\em before} any edge crossing components of $\mathcal C^{(i)}$.}
\end{array}
\end{equation}
So, $e_{jk}$ must be an edge between two components of $\mathcal C^{(i)}$. But edges between components of $\mathcal C^{(i)}$ belong to the final output $F_R$, such edges are kept instead of being deleted. This contradiction proves the claim.

\vskip 0.2cm
Now, we continue to prove property \eqref{eq1224-13}. Denote by $\mathcal C_a^{(i)}$ and $\mathcal C_{in}^{(i)}$ the set of active and inactive components of $\mathcal C^{(i)}$, respectively. Denote by $\mathcal C_{a,r}^{(i)}$ and $\mathcal C_{in,r}^{(i)}$ (resp. $\mathcal C_{a,\neg r}^{(i)}$ and $\mathcal C_{in,\neg r}^{(i)}$) the subsets of components of $\mathcal C_a^{(i)}$ and $\mathcal C_{in}^{(i)}$ containing (resp. not containing) roots, respectively. Denote by $\mathcal C_{a,d}^{(i)}$ the subset of components of $\mathcal C_a^{(i)}$ that will be contained in $C_j$ for some $j\in\{1,\ldots,q\}$ (recall that $C_1,\ldots,C_q$ are inactive components whose union is $X$, see expression \eqref{eq1224-11}).
Let the corresponding sets of contracted nodes be $N_a^{(i)},N_{in}^{(i)},N_{a,r}^{(i)},N_{in,r}^{(i)},N_{a,\neg r}^{(i)},N_{in,\neg r}^{(i)}$ and $N_{a,d}^{(i)}$, respectively.
For a node $v\in V(H)$, denote by $d_H(v)$ the degree of $v$ in $H$. An isolated node in $H$ is called a {\em trivial component} of $H$.

\vskip 0.2cm
{\bf Claim 2.} Any node $v\in N_{in,\neg r}^{(i)}$ in a non-trivial component of $H$ has $d_H(v)\geq 2$.

Since $v$ is in a non-trivial component of $H$, we have $d_H(v)\geq 1$. Suppose $d_H(v)=1$, and $e$ is the unique edge of $H$ incident with $v$. By the observation in \eqref{eq1227-1}, $e$, being an edge crossing components of $\mathcal C^{(i)}$, is added later than those edges in $\mathcal C^{(i)}$. So, every component of $\mathcal C^{(i)}$ is contained in some component of $\mathcal C_e$. Suppose $C_e$ is the component of $\mathcal C_e$ containing $C_v$. By the reverse-deletion rule, in order that $e$ is kept in the final forest $F_R$, either $C_e$ is an active component in $\mathcal C_e$, or $C_e$ contains a root, or $C_e$ is incident with some other edge in $F_R$ different from $e$. In either case, $C_v$ is a proper subset of $C_e$. This is because: in the first case, $C_e$ is active but $C_v$ is not; in the second case, $C_e$ contains a root but $C_v$ does not; in the third case, $e$ is the unique edge of $F_R$ incident with $C_v$, and thus the other edge of $F_R$ incident with $C_e$ must be incident with some vertex in $C_e-C_v$. In order to expand $C_v$ into $C_e$, some edge $e'$ is added into $F$ to merge $C_v$ with an active component $C_{e'}\in \mathcal C_{e'}$, furthermore, $e'$ is added before $e$ and $C_{e'}$ is contained in $C_e$. Since $e$ is the unique edge of final $F_R$ incident with $C_v$, the edge $e'$ must be deleted in the reverse-deletion step, together with an inactive component in $\mathcal C_{e'}$ incident with $e'$. Note that this inactive component must be $C_v$ because $C_{e'}$ is active. But then, the end of $e$ in $C_v$ is deleted, and thus $e$ cannot appear in the final $F_R$. This contradiction establishes the claim.

\vskip 0.2cm
The following argument is concentrated on the $i$th iteration, and the superscript $(i)$ is omitted to simplify the notation without causing ambiguity.


Suppose in the $i$th iteration, every active effective dual variable is increased by an amount of $\varepsilon$. Then the increase to the left-hand side of \eqref{eq108-5} is
\begin{align*}
\sum_{S\in \mathcal C_a}\varepsilon\cdot|F_R\cap\delta(S)|+2\varepsilon\cdot|\mathcal C_{a,d}|=\varepsilon\cdot\sum_{v\in N_a}d_{H}(v)+2\varepsilon\cdot|N_{a,d}|,
\end{align*}
where $d_H(v)$ is the degree of node $v$ in $H$. The increase to the right-hand side of \eqref{eq108-5} is
\begin{align*}
2\varepsilon |\mathcal C_a\setminus \mathcal C_{a,r}|= 2\varepsilon(|N_a|-|N_{a,r}|),
\end{align*}
So, to prove \eqref{eq1224-13}, it suffices to prove
\begin{align}\label{eq109-1}
	\sum_{v\in N_a}d_{H}(v)\leq 2(|N_a|-|N_{a,r}|-|N_{a,d}|).
\end{align}

To estimate the left side of \eqref{eq109-1}, we may concentrate on those nodes of $H$ whose degree is greater than zero. In the following, we use a superscript $>0$ to indicate such nodes. For example, $N_a^{>0}$ denotes the set of active nodes in $H$ whose degree is greater than zero, and let $N_a^{=0}=N_a\setminus N_a^{>0}$. The left side of \eqref{eq109-1} can be written as
\begin{equation}\label{eq1227-4}
\sum_{v\in N_a}d_{H}(v)= \sum_{v\in N_a^{>0}}d_{H}(v)= \sum_{v\in N_a^{>0}\cup N_{in}^{>0}}d_{H}(v)-\sum_{v\in N_{in}^{>0}}d_{H}(v).
\end{equation}
Denote by $V(H^{>0})$ the set of nodes belonging to non-trivial components of $H$. Note that
\begin{equation}\label{eq1227-8}
V(H^{>0})=N_a^{>0}\cup N_{in}^{>0}.
\end{equation}
Combining this with the {\em shaking-hands lemma}, an because $H^{>0}$ is a forest, we have
\begin{equation}\label{eq1227-5}
\sum_{v\in N_a^{>0}\cup N_{in}^{>0}}d_{H}(v)=\sum_{v\in V(H^{>0})}d_{H}(v)=2|E(H^{>0})|=2 (|V(H^{>0})|-\omega(H^{>0})),
\end{equation}
where $\omega(H^{>0})$ is the number of non-trivial components in $H$. By Claim 1, \begin{equation}\label{eq1227-6}
\omega(H^{>0})\geq |N_{r}^{>0}|.
\end{equation}
By Claim 2, every node $v\in N_{in,\neg r}^{>0}$ has $d_H(v)\geq 2$. So,
\begin{equation}\label{eq1227-7}
\sum_{v\in N_{in}^{>0}}d_{H}(v)=\sum_{v\in N_{in,\neg r}^{>0}}d_{H}(v)+\sum_{v\in N_{in,r}^{>0}}d_{H}(v)\geq 2|N_{in,\neg r}^{>0}|+|N_{in,r}^{>0}|.
\end{equation}
Combining expressions \eqref{eq1227-4} to \eqref{eq1227-7}, we have
\begin{align}\label{eq1228-1}
\sum_{v\in N_a}d_H(v) & \leq 2(|V(H^{>0})|-\omega(H^{>0}))-2|N_{in,\neg r}^{>0}|-|N_{in,r}^{>0}| \\ \nonumber
& \leq 2(|N_a^{>0}|+|N_{in}^{>0}|-|N_r^{>0}|-|N_{in}^{>0}|+|N_{in,r}^{>0}|)
-|N_{in,r}^{>0}|\\ \nonumber
& =2(|N_a^{>0}|-|N_r^{>0}|)+|N_{in,r}^{>0}|\\ \nonumber
& =2(|N_a^{>0}|-|N_{a,r}^{>0}|)-2|N_{in,r}^{>0}|+|N_{in,r}^{>0}|\\ \nonumber
& \leq 2(|N_a^{>0}|-|N_{a,r}^{>0}|).
\end{align}
Since all roots are spanned by the final forest $F_R$ and no node in $N_{a,d}$ is incident with any edge of $F_R$, we have
$$
N_{a,d}\subseteq N_{a,\neg r}^{=0}.
$$
It follows that
\begin{equation}\label{eq1227-2}
|N_a|-|N_{a,r}|-|N_{a,d}|\geq |N_a^{>0}|+|N_a^{=0}|-|N_{a,r}^{>0}|-|N_{a,r}^{=0}|-|N_{a,\neg r}^{=0}|=|N_a^{>0}|-|N_{a,r}^{>0}|
\end{equation}
Then inequality \eqref{eq109-1} follows from \eqref{eq1228-1} and \eqref{eq1227-2}. This finishes the proof of property \eqref{eq1224-13}, and thus \eqref{eq108-1} holds. Combining \eqref{eq108-1} and \eqref{eq1228-2}, the algorithm is a 2-LMP algorithm.

For the time complexity analysis, let us begin with the Rootless-Growth procedure. In each iteration of the while loop, we either add an edge $e$ to decrease the number of components or deactivate an active component $S$, reducing the number of active components. Initially, there are $n$ components and $n$ active components, so the while loop is executed for $O(n)$ rounds. Computing the edge $e$ and the vertex set $S$ that first become tight takes $O(m)$ and $O(n)$ time, respectively. Thus, the running time for Rootless-Growth is $O(nm)$.

Regarding the K-forest-step and the reverse-deletion step, both of these steps are executed in $O(nm)$ time. Therefore, the total running time of Algorithm 1 is $O(nm)$.
\end{proof}

In the above analysis, a degree-type inequality (i.e. \eqref{eq109-1}) plays an important role. Compared with previous works \cite{Goeman1995} and \cite{Liang-PCMSSC-TCS}, which also involve arguments on degree-type inequalities, the challenge in this paper lies in the estimation of $|N_{a,r}|$. This challenge is caused by ``rootless growth'': for a forest obtained by a rooted-growth method, a component containing a root always remains inactive, which leads to $|N_{a,r}|=0$, while for the forest obtained by a rootless-growth method, $|N_{a,r}|$ might not be zero. A property that guarantees a small value of $|N_{a,r}|$ is \eqref{eq1227-6}. Although it might be natural to break a component with more than one roots by deleting the heaviest edge on the path connecting two roots, such a method cannot guarantee property \eqref{eq1227-6}. While our method of deleting ``the last added edge'' has property \eqref{eq1227-1}, which leads to the validity of Claim 1, and in turn inequality \eqref{eq1227-6}.

\subsection{$\mathcal A_1$: Rootless-Growth+Rootless-Prune}\label{sec3-3}

In the outline ideas presented at the beginning of Section \ref{sec002}, we construct a rooted RPCF$_K$ instance $IR$ from an unrooted URPCF$_K$ instance $I$, ensuring $opt_{I_R}=opt_I$. Theorem \ref{theo1} demonstrates that Algorithm \ref{algo1} is able to compute a 2-LMP solution for $IR$. However, both the $K$-forest-step and the reverse-deletion step rely on prior knowledge of roots, sourced from an optimal solution of instance $I$, resulting in a problematic loop.

In this subsection, we replace these steps, collectively termed the ``rooted pruning scheme'', with a ``rootless pruning scheme''. By combining rootless-growth and rootless-prune methods, we can attain a 2-LMP solution for instance $I$ without assuming any information about roots or the optimal solution. We will show that the rootless pruning scheme is at least as effective as the rooted pruning scheme, ensuring no deterioration in performance. That is,
\begin{equation}\label{eq0703-1}
	w(\mathcal{F}_{UR})+2\pi(V\setminus V(\mathcal{F}_{UR}))\leq w(\mathcal{F}_{R})+2\pi(V\setminus V(\mathcal{F}_{R})),
\end{equation}
where, $\mathcal F_{UR}$ is the set of $K$ subtrees after the rootless pruning scheme and $\mathcal F_R$ is the set of $K$ subtrees after the rooted pruning scheme.

We employ a concept known as ``net worth'' to assist the rootless prune scheme, initially introduced in \cite{Johnson2000} as a bridge to solve the prize-collecting Steiner tree problem. Here, we have adapted it to serve our specific purpose in conducting an LMP analysis for a forest problem.
For a forest $\mathcal{T}$ and a sub-forest $\mathcal T'\subseteq \mathcal T$, define the {\em net worth value of $\mathcal T'$ with respect to $\mathcal T$}, denoted as $nw_\mathcal{T}(\mathcal{T'})$, as
$$nw_\mathcal{T}(\mathcal{T'})=2\pi(V(\mathcal{T}'))-w(\mathcal{T}').$$
and the {\em prize-collecting value} of $\mathcal T'$ with respect to $\mathcal{T}$, denoted as $pc_{\mathcal{T}}(\mathcal{T}')$, as $$pc_{\mathcal{T}}(\mathcal{T}')=w(\mathcal{T}')+2 \pi(V(\mathcal{T})\setminus V(\mathcal{T}')).$$
Note that $pc_{\mathcal{T}}(\mathcal{T}')+nw_{\mathcal{T}}(\mathcal{T}')=2\pi(V(\mathcal{T})))$, which is a constant determined by $\mathcal{T}$. So, finding a subforest $\mathcal{T}'$ of $\mathcal{T}$ with the minimum $pc_{\mathcal{T}}$-values is equivalent to finding such a subforest with the maximum $nw_{\mathcal{T}}$-value.

Recall that $\mathcal C$ represents the set of connected components at the end of the Rootless-Growth procedure applied to instance $I$, while $\mathcal C_R$ represents the set of connected components at the end of the Rootless-Growth procedure applied to instance $IR$. It is important to note that the rootless-growth step is not influenced by the presence of roots, and therefore $\mathcal C = \mathcal C_R$. The only distinction between $\mathcal C$ and $\mathcal C_R$ is that $\mathcal C_R$ includes certain predetermined roots that are assumed to be known, even though they are not explicitly given in the instance. Therefore, our goal \eqref{eq0703-1} is equal to
\begin{equation}\label{eq0104-1}
	nw_{\mathcal{C}}(\mathcal{F}_{UR})\geq nw_{\mathcal{C}_R}(\mathcal{F}_R).
\end{equation}

To implement the ideas outlined above, we need to address two key challenges. The first challenge is, given a tree $T\in\mathcal C$ and an integer $k$, how to prune $T$ into $k$ disjoint subtrees such that the performance ratio 2 is still kept. Call this procedure as {\em rootless-prune realization}, write it as Rootless-Prune($T,k$) and denote its output as $F_{T,k}$. The second challenge is, for each tree $T\in\mathcal{C}$, how to determine the number $k_T$ of connected components that $T$ should be broken into. Call a parameter set $\{k_T:T\in\mathcal{C}\}$ satisfying $\sum_{T\in\mathcal{C}}k_T=K$ as a {\em configuration-scheme}. The {\em best-configuration-scheme} is a configuration scheme $\{k_T^*:T\in\mathcal{C}\}$ with
\begin{equation}\label{eq0613-1}
	\{k_T^*:T\in\mathcal{C}\}=\arg\max_{\{k_T:T\in\mathcal{C}\}}\sum_{T\in\mathcal{C}} nw_T(F_{T,k_T}).
\end{equation}
It is interesting to see that we can make use the algorithm Rootless-Prune$(T^{aux},K)$ on an auxiliary tree $T^{aux}$ to directly find out a forest $\mathcal F_{UR}$ satisfying property \eqref{eq0104-1}, and the best-configuration-scheme is just a by-product.

Section \ref{sec0107-1} gives the procedure Rootless-Prune$(T,k)$. Section \ref{sec0107-2} presents the algorithm to compute a forest $\mathcal F_{UR}$ satisfying property \eqref{eq0104-1}.


\subsubsection{Rootless-Prune($T,k$)}\label{sec0107-1}

Given a tree $T$ and an integer $k$, in this subsection, we show how to prune $T$ into a forest $F_{T,k}$ consisting of $k$ subtrees that has the largest net-worth value.

If $k=0$, then trivially set $F_{T,0}=\emptyset$.
%
%
For $k\geq 1$, define an {\em NW-maximum $k$-forest} of $T$ as a forest $F_{T,k}^*$ composed of $k$ disjoint subtrees of $T$ having the maximum $nw_T(F_{T,k}^*)$. Note that \cite{Johnson2000} described a greedy method to find an NW-maximum subtree (or in our terminology, an NW-maximum 1-forest) of $T$. In the following, we shall use a double dynamic programming method to find an NW-maximum $k$-forest for any $1\leq k\leq n$. In the following, we omit the subscript $T$ of $nw_T$ for simplicity.

\vskip 0.2cm
{\bf A natural method cannot yield an NW-Maximum Forest.} To find an NW-maximum $k$-forest, a logical approach is to start with an NW-maximum 1-forest $F^*_{T,1}$ of tree T, and then devise an iterative algorithm to derive a $k$-forest $F_{T,k}$ from a $(k - 1)$-forest $F_{T,k-1}$, which operates in a greedy manner, aiming to maximize the net worth value increase at each step. There are two natural ways to obtain $F_{T,k}$ from $F_{T,k-1}$:

$(a)$ One way is to add a subtree from $T-F_{T,k-1}$ into $F_{T,k-1}$.

$(b)$ The other way is to break a subtree of $F_{T,k-1}$ into two subtrees.

\begin{figure}[H]
	\begin{center}
		\tikzset{every picture/.style={line width=1pt}} 
		
		\begin{tikzpicture}[x=0.75pt,y=0.75pt,yscale=-1,xscale=1]
			
			\draw    (198.75,207.25) -- (252.25,107.25) ;
			\draw [shift={(252.25,107.25)}, rotate = 298.15] [color={rgb, 255:red, 0; green, 0; blue, 0 }  ][fill={rgb, 255:red, 0; green, 0; blue, 0 }  ][line width=0.75]      (0, 0) circle [x radius= 1.34, y radius= 1.34]   ;
			\draw [shift={(198.75,207.25)}, rotate = 298.15] [color={rgb, 255:red, 0; green, 0; blue, 0 }  ][fill={rgb, 255:red, 0; green, 0; blue, 0 }  ][line width=0.75]      (0, 0) circle [x radius= 1.34, y radius= 1.34]   ;
			\draw    (233.75,206.75) -- (252.25,107.25) ;
			\draw [shift={(233.75,206.75)}, rotate = 280.53] [color={rgb, 255:red, 0; green, 0; blue, 0 }  ][fill={rgb, 255:red, 0; green, 0; blue, 0 }  ][line width=0.75]      (0, 0) circle [x radius= 1.34, y radius= 1.34]   ;
			\draw    (269.75,206.75) -- (252.25,107.25) ;
			\draw [shift={(269.75,206.75)}, rotate = 260.02] [color={rgb, 255:red, 0; green, 0; blue, 0 }  ][fill={rgb, 255:red, 0; green, 0; blue, 0 }  ][line width=0.75]      (0, 0) circle [x radius= 1.34, y radius= 1.34]   ;
			\draw    (305.25,207.75) -- (252.25,107.25) ;
			\draw [shift={(305.25,207.75)}, rotate = 242.19] [color={rgb, 255:red, 0; green, 0; blue, 0 }  ][fill={rgb, 255:red, 0; green, 0; blue, 0 }  ][line width=0.75]      (0, 0) circle [x radius= 1.34, y radius= 1.34]   ;
			
			\draw (257.5,96.13) node {$2$};
			\draw (265.5,215.63) node{$7.5$};
			\draw (230,215.63) node{$6.5$};
			\draw (198,216.63) node{$6.2$};
			\draw (210,165.13) node{$9$};
			\draw (293.5,165.13) node{$11$};
			\draw (271.5,166.63) node{$9$};
			\draw (236.5,166.13) node{$9$};
			\draw (303,216.13) node{$5.5$};
			\draw (251.5,250) node{\scriptsize(a)\ The tree $T$ and suppose $k=3$};

		\end{tikzpicture}
		\hskip 0.8cm
		\begin{tikzpicture}[x=0.75pt,y=0.75pt,yscale=-1,xscale=1]
			
			\draw [line width=2]    (198.75,207.25) -- (252.25,107.25) ;
			\draw [shift={(252.25,107.25)}, rotate = 298.15] [color={rgb, 255:red, 0; green, 0; blue, 0 }  ][fill={rgb, 255:red, 0; green, 0; blue, 0 }  ][line width=2]      (0, 0) circle [x radius= 2.61, y radius= 2.61]   ;
			\draw [shift={(198.75,207.25)}, rotate = 298.15] [color={rgb, 255:red, 0; green, 0; blue, 0 }  ][fill={rgb, 255:red, 0; green, 0; blue, 0 }  ][line width=2]      (0, 0) circle [x radius= 2.61, y radius= 2.61]   ;
			\draw [line width=2]    (233.75,206.75) -- (252.25,107.25) ;
			\draw [shift={(233.75,206.75)}, rotate = 280.53] [color={rgb, 255:red, 0; green, 0; blue, 0 }  ][fill={rgb, 255:red, 0; green, 0; blue, 0 }  ][line width=2]      (0, 0) circle [x radius= 2.61, y radius= 2.61]   ;
			\draw [line width=2]    (269.75,206.75) -- (252.25,107.25) ;
			\draw [shift={(269.75,206.75)}, rotate = 260.02] [color={rgb, 255:red, 0; green, 0; blue, 0 }  ][fill={rgb, 255:red, 0; green, 0; blue, 0 }  ][line width=2]      (0, 0) circle [x radius= 2.61, y radius= 2.61]   ;
			\draw [line width=2]    (305.25,207.75) -- (252.25,107.25) ;
			\draw [shift={(305.25,207.75)}, rotate = 242.19] [color={rgb, 255:red, 0; green, 0; blue, 0 }  ][fill={rgb, 255:red, 0; green, 0; blue, 0 }  ][line width=2]      (0, 0) circle [x radius= 2.61, y radius= 2.61]   ;
			\draw (257.5,96.13) node {$2$};
			\draw (265.5,218) node{$7.5$};
			\draw (230,218) node{$6.5$};
			\draw (198,218) node{$6.2$};
			\draw (210,165.13) node{$9$};
			\draw (293.5,165.13) node{$11$};
			\draw (271.5,166.63) node{$9$};
			\draw (236.5,166.13) node{$9$};
			\draw (303,218) node{$5.5$};
			\draw (251.5,250) node{\scriptsize (b)\ $nw(F^*_{T,1})=nw(T)=17.4$};
			
		\end{tikzpicture}
		\hskip 0.8cm
		\begin{tikzpicture}[x=0.75pt,y=0.75pt,yscale=-1,xscale=1]
			
			\draw [line width=2]    (198.75,207.25) -- (252.25,107.25) ;
			\draw [shift={(252.25,107.25)}, rotate = 298.15] [color={rgb, 255:red, 0; green, 0; blue, 0 }  ][fill={rgb, 255:red, 0; green, 0; blue, 0 }  ][line width=2]      (0, 0) circle [x radius= 2.61, y radius= 2.61]   ;
			\draw [shift={(198.75,207.25)}, rotate = 298.15] [color={rgb, 255:red, 0; green, 0; blue, 0 }  ][fill={rgb, 255:red, 0; green, 0; blue, 0 }  ][line width=2]      (0, 0) circle [x radius= 2.61, y radius= 2.61]   ;
			\draw [line width=2]    (233.75,206.75) -- (252.25,107.25) ;
			\draw [shift={(233.75,206.75)}, rotate = 280.53] [color={rgb, 255:red, 0; green, 0; blue, 0 }  ][fill={rgb, 255:red, 0; green, 0; blue, 0 }  ][line width=2]      (0, 0) circle [x radius= 2.61, y radius= 2.61]   ;
			\draw [line width=2]    (269.75,206.75) -- (252.25,107.25) ;
			\draw [shift={(269.75,206.75)}, rotate = 260.02] [color={rgb, 255:red, 0; green, 0; blue, 0 }  ][fill={rgb, 255:red, 0; green, 0; blue, 0 }  ][line width=2]      (0, 0) circle [x radius= 2.61, y radius= 2.61]   ;
			\draw [line width=0.75,dashed]    (305.25,207.75) -- (252.25,107.25) ;
			\draw [shift={(305.25,207.75)}, rotate = 242.19] [color={rgb, 255:red, 0; green, 0; blue, 0 }  ][fill={rgb, 255:red, 0; green, 0; blue, 0 }  ][line width=2]      (0, 0) circle [x radius= 3.02, y radius= 3.02]   ;
			
			\draw (257.5,96.13) node {$2$};
			\draw (265.5,218) node{$7.5$};
			\draw (230,218) node{$6.5$};
			\draw (198,218) node{$6.2$};
			\draw (210,165.13) node{$9$};
			\draw (293.5,165.13) node{$11$};
			\draw (271.5,166.63) node{$9$};
			\draw (236.5,166.13) node{$9$};
			\draw (303,218) node{$5.5$};
			\draw (251.5,250) node{\scriptsize(c) \ $nw(F_{T,2})=nw(F^*_{T,1})+11=28.4$};
		\end{tikzpicture}
		\vskip 0.5cm
		\begin{tikzpicture}[x=0.75pt,y=0.75pt,yscale=-1,xscale=1]
			
			\draw [line width=0.75,dashed]    (198.75,207.25) -- (252.25,107.25) ;
			\draw [shift={(252.25,107.25)}, rotate = 298.15] [color={rgb, 255:red, 0; green, 0; blue, 0 }  ][fill={rgb, 255:red, 0; green, 0; blue, 0 }  ][line width=0.75]      (0, 0) circle [x radius= 1.34, y radius= 1.34]   ;
			\draw [shift={(198.75,207.25)}, rotate = 298.15] [color={rgb, 255:red, 0; green, 0; blue, 0 }  ][fill={rgb, 255:red, 0; green, 0; blue, 0 }  ][line width=0.75]      (0, 0) circle [x radius= 1.34, y radius= 1.34]   ;
			\draw [line width=0.75,dashed]    (233.75,206.75) -- (252.25,107.25) ;
			\draw [shift={(233.75,206.75)}, rotate = 280.53] [color={rgb, 255:red, 0; green, 0; blue, 0 }  ][fill={rgb, 255:red, 0; green, 0; blue, 0 }  ][line width=2]      (0, 0) circle [x radius= 2.61, y radius= 2.61]   ;
			\draw [line width=0.75,dashed]    (269.75,206.75) -- (252.25,107.25) ;
			\draw [shift={(269.75,206.75)}, rotate = 260.02] [color={rgb, 255:red, 0; green, 0; blue, 0 }  ][fill={rgb, 255:red, 0; green, 0; blue, 0 }  ][line width=2]      (0, 0) circle [x radius= 2.61, y radius= 2.61]   ;
			\draw [shift={(305.25,207.75)}, rotate = 242.19] [color={rgb, 255:red, 0; green, 0; blue, 0 }  ][fill={rgb, 255:red, 0; green, 0; blue, 0 }  ][line width=2]      (0, 0) circle [x radius= 3.02, y radius= 3.02]   ;
			
			\draw (257.5,96.13) node {$2$};
			\draw (265.5,218) node{$7.5$};
			\draw (230,218) node{$6.5$};
			\draw (198,218) node{$6.2$};
			\draw (210,165.13) node{$9$};
			\draw (271.5,166.63) node{$9$};
			\draw (236.5,166.13) node{$9$};
			\draw (303,218) node{$5.5$};
			\draw (251.5,250) node{\scriptsize(d)\ $nw(F_{T,3})=nw(F_{T,2})+10.6=39$};
		\end{tikzpicture}
		\hskip 0.5cm
		\begin{tikzpicture}[x=0.75pt,y=0.75pt,yscale=-1,xscale=1]
			
			\draw [line width=0.75]    (198.75,207.25) -- (252.25,107.25) ;
			\draw [shift={(252.25,107.25)}, rotate = 298.15] [color={rgb, 255:red, 0; green, 0; blue, 0 }  ][fill={rgb, 255:red, 0; green, 0; blue, 0 }  ][line width=0.75]      (0, 0) circle [x radius= 1.34, y radius= 1.34]   ;
			\draw [shift={(198.75,207.25)}, rotate = 298.15] [color={rgb, 255:red, 0; green, 0; blue, 0 }  ][fill={rgb, 255:red, 0; green, 0; blue, 0 }  ][line width=2]      (0, 0) circle [x radius= 2.61, y radius= 2.61]   ;
			\draw [line width=0.75]    (233.75,206.75) -- (252.25,107.25) ;
			\draw [shift={(233.75,206.75)}, rotate = 280.53] [color={rgb, 255:red, 0; green, 0; blue, 0 }  ][fill={rgb, 255:red, 0; green, 0; blue, 0 }  ][line width=2]      (0, 0) circle [x radius= 2.61, y radius= 2.61]   ;
			\draw [line width=0.75]    (269.75,206.75) -- (252.25,107.25) ;
			\draw [shift={(269.75,206.75)}, rotate = 260.02] [color={rgb, 255:red, 0; green, 0; blue, 0 }  ][fill={rgb, 255:red, 0; green, 0; blue, 0 }  ][line width=2]      (0, 0) circle [x radius= 2.61, y radius= 2.61]   ;
			\draw [line width=0.75]    (305.25,207.75) -- (252.25,107.25) ;
			\draw [shift={(305.25,207.75)}, rotate = 242.19] [color={rgb, 255:red, 0; green, 0; blue, 0 }  ][fill={rgb, 255:red, 0; green, 0; blue, 0 }  ][line width=0.75]      (0, 0) circle [x radius= 1.34, y radius= 1.34]   ;
			
			\draw (257.5,96.13) node {$2$};
			\draw (265.5,218) node{$7.5$};
			\draw (230,218) node{$6.5$};
			\draw (198,218) node{$6.2$};
			\draw (210,165.13) node{$9$};
			\draw (293.5,165.13) node{$11$};
			\draw (271.5,166.63) node{$9$};
			\draw (236.5,166.13) node{$9$};
			\draw (303,218) node{$5.5$};
			\draw (251.5,250) node{\scriptsize(e) \ $nw(F^*_{T,3})=12.4+13+15= 40.4$};
		\end{tikzpicture}
	\end{center}
	\caption{An example showing that the greedy idea cannot yield optimal $k$-forest.}\label{fig2}
\end{figure}
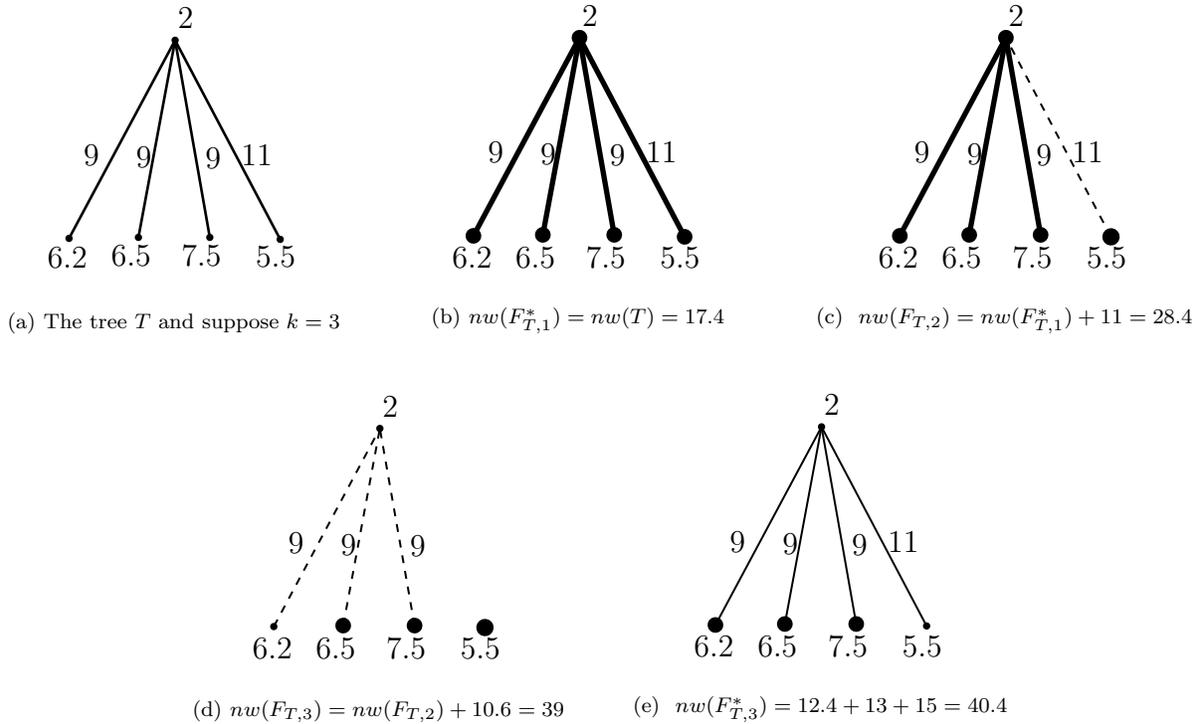

However, this method cannot compute an optimal solution, as indicated by the example in Fig.~\ref{fig2}. In this example, the vertex penalties and the edge weights are marked by the numbers beside them, and for convenience, these numbers are also used to refer to the corresponding vertices and edges. The goal is to construct a 3-forest. In Fig.~\ref{fig2} (b), the blackened lines and blackened nodes indicate the NW-maximum 1-forest $F_{T,1}^*$ of $T$ computed in the first iteration, achieving net wroth value 17.4. Fig.~\ref{fig2} (c) depicts the result of the second iteration: deleting the dashed edge from $F_{T,1}^*$ results in a 2-forest $F_{T,2}$ with the biggest increase of net worth value, such that $nw(F_{T,2})=17+11=28$. Fig.~\ref{fig2} (d) depicts the result of the last iteration: deleting from $F_{T,2}$ the three dashed edges and the two vertices 2 and 6 results in a 3-forest $F_{T,3}$ with the biggest increase of net worth value, such that $nw(F_{T,3})=39$. However, the NW-maximum 3-forest $F^*_{T,3}$, indicated by Fig.~\ref{fig2} (e), is composed of three isolated vertices $\{6.2,6.5,7.5\}$ with $nw(F^*_{T,3})=40.4$.

In the following, we design a {\em double dynamic programming} method to obtain  an NW-maximum $k$-forest. The reason why we call it double dynamic is because we shall embed two (sub-)dynamic algorithms in a main dynamic programming algorithm.

\vskip 0.2cm
{\bf A double dynamic programming.} For simplicity of statement, we choose an arbitrary vertex $r\in V(T)$ as the root and denote the rooted tree as $T_r$. For a vertex $u\in V(T_r)$, let
$CH(u)=\{v\in V(T_r):v \mbox{ is a child } u\}$ be the set of children of $u$.
Define the status variables as follows. Let $F_u=F(u,p_u,t_u,k_u)$ be a maximum net-worth sub-forest of $T_u$ satisfying the following conditions: $p_u=1$ indicates that $u$ belongs to this sub-forest, otherwise $p_u=0$; $t_u=1$ indicates that $u$ is a root, otherwise $t_u=0$; there are exactly $k_u$ roots in this sub-forest ($t_u$ and $k_u$ are combined to help counting the number of components of the final forest). Note that if $p_u=0$, then $t_u=0$. If $F'$ is a forest satisfying the above conditions but may not have the maximum net worth value, then we say that $F'$ is a {\em feasible forest for tuple $(u,p_u,t_u,k_u)$.} Let $f(u,p_u,t_u,k_u)$ be the net-worth of $F(u,p_u,t_u,k_u)$. The transition formula for the computation of $f(u,p_u,t_u,k_u)$ is presented in the following lemma.

\begin{lemma}\label{eq0406-1}
	{\rm The value $f(u,p_u,t_u,k_u)$ can be computed by the following formula:
		\begin{equation*}
			f(u,p_u,t_u,k_u)=\left\{\begin{array}{ll}
				\max\limits_{\{k_v\}_{v\in CH(u)}:\atop\sum_{v\in CH(u)} k_v=k_u}\left\{\sum\limits_{v\in CH(u)\atop {\mbox{\tiny if $p_v=1$ then $t_v=1$}}}f(v,p_v,t_v,k_v)\right\},& \mbox{if $p_u=0$},\\
				\\
				\max\limits_{S(u)\subseteq CH(u), \atop {\{k_v\}_{v\in CH(u)}:\atop \sum_{v\in CH(u)} k_v=k_u-t_u}}\left\{\begin{array}{ll}
					\sum\limits_{v\in S(u)}(f(v,1,0,k_v)-w(uv)) \\
					+ \sum\limits_{v\in CH(u)-S(u)\atop {\mbox{\tiny if $p_v=1$ then $t_v=1$}}} f(v,p_v,t_v,k_v)\\
					+2\pi(u)
				\end{array}\right\},& \mbox{if $p_u=1$}.
			\end{array}\right.
		\end{equation*}
		The initial conditions are as follows.
		
		$(\romannumeral1)$ For any vertex $u$ of $T_r$, $f(u,0,1,k_u)=-\infty$;
		
		$(\romannumeral2)$ For a leaf $u$ of $T_r$, $f(u,1,1,1)=f(u,1,0,0)=2\pi(u)$, $f(u,0,0,0)=0$, and for the other 4-tuples, $f(u,p_u,t_u,k_u)=-\infty$.}
	%
\end{lemma}

The ideas behind the transition formula are stated as follows.

\begin{itemize}
	\item {\bf Case 1}: $p_u=0$ \\
	Note that $p_u=0$ implies $t_u=0$. In this case, for every $v\in CH(u)$, $F'_v=F(u,p_u,t_u,k_u)\cap T_v$ is a feasible forest for some tuple $(v,p_v,t_v,k_v)$. If $v\in F'_v$ (i.e. $p_v=1$), then $v$ must be a root (i.e. $t_v=1$) because $u\not\in F(u,p_u,t_u,k_u)$. The forests $\{F'_v\}_{v\in CH(u)}$ are independent with each other (because $p_u=0$) and collectively contains $k_u$ roots, that is $\sum_{v\in CH(u)} k_v=k_u$. In order that $F(u,p_u,t_u,k_u)$ has the maximum net worth value, $F'_v$ must have the maximum net worth value among all feasible forests for tuple $(v,p_v,t_v,k_v)$, that is, $F'_v=F(v,p_v,t_v,k_v)$. Then $F(u,p_u,t_u,k_u)$ is the best one among the union of $\{F(v,p_v,t_v,k_v)\}_{v\in CH(u)}$ conforming to these conditions.
	\item {\bf Case 2}: $p_u=1$\\
	Note that $p_u=1$ means $u\in F(u,p_u,t_u,k_u)$. Similar to Case 1, $F_v'$ is used to denote $F'_v=F(u,p_u,t_u,k_u)\cap T_v$. If $u$ is a root (i.e., $t_u=1$), then forest $\{F'_v:v\in CH(u)\}$ collectively contains $k_u-1$ roots, otherwise $\{F'_v:v\in CH(u)\}$ collectively contains $k_u$ roots. Compared with Case 1, besides the allocation $\{k_v\}_{v\in CH(u)}$, we have to determine the set $S(u)$ that consists of those children of $u$ linked to $u$ in $F(u,p_u,t_u,k_u)$. For $v\in S(u)$, we have $p_v=1$ and $t_v=0$. The net worth value of such a forest can be computed by the formula in the second case of the transition formula, with $\pi(u)$ and $\{w(uv)\}_{v\in S(u)}$ being taken into consideration.
\end{itemize}

Note that in the above transition formula, we have to consider the allocation of $\{k_v\}_{v\in CH(u)}$ satisfying $\sum_{v\in CH(u)} k_v=k_u$, which may take exponential time if a direct enumeration method is used. In the following, we show that there is a cleverer method to realize this task by embedding two (sub-)dynamic programming algorithms in the above (main-)dynamic programming formula.

For convenience, we say that $p_v$ and $t_v$ are compatible if $p_v=1$ implies $t_v=1$ in the first case and $p_v=1$ implies $t_v=0$ for $v\in S(u)$ in the second case. Suppose $u$ has $d_u$ children $\{v_1,\ldots,v_{d_u}\}$. Denote by $CH(u,i)=\{v_1,\ldots,v_i\}$.

\vskip 0.2cm
{\bf Embedded DP for Case 1.} For $i=1,\ldots,d_u$ and an integer $K_i\in\{0,1,\ldots,k_u\}$, let $\widetilde{F}(i,K_i)$ be a maximum net-worth sub-forest of $\bigcup_{v\in CH(u,i)}T_v$ satisfying the following conditions: $\widetilde F(i,K_i)=\bigcup_{v\in CH(u,i)}F(v,p_v,t_v,k_v)$ for some compatible pairs $\{(p_v,t_v)\}_{v\in CH(u,i)}$ and some integers $\{k_v\}_{v\in CH(u,i)}$ with $\sum_{v\in CH(u,i)}k_v=K_i$ and each $k_v\in \{0,1,\ldots,K_i\}$. Let $\widetilde f(i,K_i)=nw(\widetilde F(i,K_i))$.
The following lemma shows how to compute $\widetilde f(i,K_i)$.
%
%
%

\begin{lemma}\label{lem0215-1}
	{\rm
		The value $\widetilde f(i,K_i)$ can be computed by the following transition formula:
		\begin{align*}
			\widetilde f(i,K_i)=\max_{k_{v_i}\in\{0,1,\ldots,n_i\}\atop {p_{v_i},t_{v_{i}}\in\{0,1\}\atop{t_{v_i}, p_{v_i} \mbox{\tiny are compatible}}}} \left\{\widetilde f(i-1,K_i-k_{v_i})+f(v_i,p_{v_i},t_{v_i},k_{v_i})\right\},
		\end{align*}
		where $n_i=\min\{K_i,|V(T_{v_i})|\}$. The initial conditions are: if $0\leq K_1\leq n_1$ then $\widetilde f(1,K_1)=\max \{f(v_1,p_{v_1},t_{v_1},K_1)\}$, where the maximum is taken over all pairs $(p_v,t_v)$ that are compatible; if $K_1>n_1$ then $\widetilde f(1,K_1)=-\infty$.}
\end{lemma}

\begin{coro}\label{coro0612-3}
	{\rm The value $f(u,p_u,t_u,k_u)$ in Lemma \ref{eq0406-1} for $p_u=0$ can be computed using Lemma \ref{lem0215-1} in time $O(d_u(k_u+1)^2)$.}
\end{coro}
\begin{proof}
	Note that if $p_u=0$ implies $t_u=0$ and by the definitions of $f$ and $\widetilde f$, we have
	$$
	f(u,0,0,k_u)=\widetilde f(d_u,k_u).
	$$
	So, to compute the $f$-value for $p_u=0$, we may use Lemma \ref{lem0215-1} to compute $\{\widetilde f(i,K_i)\colon i=1,\ldots,d_u;K_i=0,1,\ldots,k_u\}$ and output the last value $\widetilde f(d_u,k_u)$. The time complexity follows from the observations that there are $d_u(k_u+1)$ values of the form $\widetilde f(i,K_i)$ to be computed and computing each $\widetilde f(i,K_i)$ takes times $2(n_i+1)\leq 2(k_u+1)$.
\end{proof}

{\bf Embedded DP for Case 2.} Compared with Case 1, we have to additionally determine set $S(u)$.  For this purpose, for $i=1,\ldots,d$, we add a variable $q_i$ to indicate whether $v_i$ is linked to $u$ or not, and add variable $s_i$ to record how many children in $CH(u,i)$ have to be linked to $u$. Let $T_u^{(i)}=u+\{uv\}_{v\in CH(u,i)}+\{T_v\}_{v\in CH(u,i)}$ be the subtree of $T_u$ restricted to its first $i$ children. For $i\in\{1,\ldots,d\}$, $q_i\in\{0,1\}$, $K_i\in\{t_u,t_u+1,\ldots,k_u\}$ and $s_i\in\{0,1,\ldots,i\}$, define $\widetilde{F}(t_u,i,q_i,K_i,s_i)$ to be a maximum net-worth sub-forest of $T_u^{(i)}$ with the following structure: this forest contains $K_i$ roots; if $q_i=1$ then edge $uv_i$ belongs to this forest, otherwise $uv_i$ does not belong to this forest; there are exactly $s_i$ children in $CH(u,i)$ that are linked to $u$, that is, there exists a set $S(u,i)\subseteq CH(u,i)$ with $|S(u,i)|=s_i$ such that
$$
\widetilde{F}(t_u,i,q_i,K_i,s_i)=u+\{uv\}_{v\in S(u,i)}+\bigcup_{v\in S(u,i)}F(v,1,0,k_v)+\bigcup_{CH(u,i)-S(u,i)}F(v,p_v,t_v,k_v)
$$
for some compatible pairs $\{(p_v,t_v)\}_{v\in CH(u,i)-S(u,i)}$ and some integers $\{k_v\}_{v\in CH(u,i)}$ satisfying $\sum_{v\in CH(u,i)}k_v=K_i-t_u$ with $k_v\in \{0,1,\ldots,K_i-t_u\}$. Let $\widetilde{f}(t_u,i,q_i,K_i,s_i)=nw(\widetilde{F}(t_u,i,q_i,K_i,s_i))$. The following lemma shows how to compute $\widetilde f(t_u,i,q_i,K_i,s_i)$.

\begin{lemma}\label{lem0412-1}
	{\rm Denote $n_i=\min\{K_i,|V(T_{v_i})|\}$. Let
		$$
		\hat{f}(i,q_i,k_{v_i})=\left\{\begin{array}{ll} f(v_i,1,0,k_{v_i})-w(uv_i), & q_i=1,\\ \max\{f(v_i,1,1,k_{v_i}),f(v_i,0,0,k_{v_i})\}, & q_i=0.\end{array}\right.
		$$
		The value $\widetilde{f}(t_u,i,q_i,K_i,s_i)$ can be computed by the transition formula:	
		
		\begin{align*}
			\widetilde{f}(t_u,i,q_i,K_i,s_i)=\max_{k_{v_i}\in\{0,1,\ldots,n_i\}
				\atop{q_{i}\in\{0,1\}}}\left\{\widetilde{f}\left(t_u,i-1,q_{i-1},K_i-k_{v_i},s_i-q_i\right)+ \hat{f}(i,q_i,k_{v_i})\right\},
		\end{align*}
		where $n_i=\min\{K_i,|V(T_{v_i})|\}$. The initial conditions are: if $0\leq K_1\leq n_1$, then $\widetilde f(t_u,1,0,K_1,0)=\hat f(1,0,K_1-t_u)$, $\widetilde f(t_u,1,1,K_1,1)=\hat f(1,1,K_1-t_u)$ and all the other $\widetilde f(t_u,1,q_1,K_1,s_1)=-\infty$.}
\end{lemma}

\begin{coro}\label{coro0612-4}
	{\rm The value $f(u,p_u,t_u,k_u)$ in Lemma \ref{eq0406-1} for $p_u=1$ can be computed using Lemma \ref{lem0412-1} in time ¡¢ $O((d_u+1)^2(k_u+1)^2)$.}
\end{coro}

\begin{proof}
	Note that for $p_u=1$, by the definitions of $f$ and $\widetilde f$,
	$$
	f(u,1,t_u,k_u)=\max_{q_{d_u}\in\{0,1\}\atop{ s_{d_u}\in\{0,1\,\ldots,d_u\}}}\{\widetilde f( t_u,d_u,q_{d_u},k_u-t_u, s_{d_u})\}.
	$$
	Denote by $q_u^*,s_u^*$ the parameter $q_{d_u},s_{d_u}$ that achieve the above maximum.
	So, to compute the $f$-value for $p_u=1$, we may use Lemma \ref{lem0215-1} to compute $\{\widetilde f( t_u,i,q_i,K_i,s_i)\colon$ $i=1,\ldots,d_u;\ q_i=0,1;\ K_i=t_u,t_u+1,\ldots,k_u;\ s_i=0,1,\ldots,i\}$ and output the last value $\widetilde f(t_u,d_u,q_u^*,k_u-t_u,s_u^*)$. The time complexity follows from the observations that there are $d_u(d_u+1)(k_u-t_u+1)$ values of the form $\widetilde f(t_u,i,q_i,K_i,s_i)$ to be computed and computing each $\widetilde f(t_u,i,q_i,K_i,s_i)$ takes times $2(n_i+1)\leq 2(k_u+1)$.
\end{proof}

\vskip 0.2cm
{\bf Combine the above together.} The pseudo code for the double dynamic programming method is presented in Algorithm \ref{algo4}

\begin{algorithm}[H]
	\caption {Rootless-Prune($T,k$)}
	\hspace*{0.02in}\raggedright{\bf Input:} A tree $T_r$ with an arbitrary root $r\in V(T)$ and an integer $0\leq k\leq |V(T)|$,
	
	\hspace*{0.02in}\raggedright{\bf Output:} A forest $F_{T,k}\subseteq T$ with $k$ connected components.
	
	\begin{algorithmic}[1]
		\STATE $F_{T,k}\leftarrow$ use Lemma~\ref{eq0406-1} embedded with Lemma \ref{lem0215-1} to compute $\max\limits_{p_r,t_r\in\{0,1\}}f(r,p_r,t_r,k)$;
		\RETURN $F_{T,k}$
	\end{algorithmic}\label{algo4}
\end{algorithm}

\subsubsection{Computation of forest $\mathcal{F}_{UR}$ satisfying property \eqref{eq0104-1}}\label{sec0107-2}
Construct an auxiliary rooted tree $T_{r'}^{aux}$ as follows. Let $\mathcal{C}$ be the forest computed by rootless-growth; create a dummy vertex $r'$ with an arbitrary penalty $\pi(r')$; for each tree $T\in\mathcal{C}$, choose an arbitrary vertex $r_T$ of $T$ as its root, link $r_T$ and $r'$ by an edge $\{r_{T}r'\}$ with an arbitrary weight $w(r_{T}r')$, resulting in a tree $T_{r'}^{aux}$ rooted at $r'$. Run the algorithm Rootless-Prune on input ($T_{r'}^{aux},K$). Note that $f(r',0,0,K)$ has the maximum net-worth value among all sub-forests of $\bigcup_{T\in\mathcal C}T$ containing exactly $K$ roots (and thus exactly $K$ components). So, $F(r',0,0,K)$ is the desired forest $\mathcal{F}_{UR}$ satisfying property \eqref{eq0104-1}.

As a by-product, for each $T\in\mathcal C$, let $k_T$ equal the number of roots of $F(r',0,0,K)$ contained in $T$, then $\{k_T:T\in\mathcal{C}\}$ is a best-configuration-scheme defined in \eqref{eq0613-1}.
The complete algorithm for URPCF$_K$ is presented in Algorithm~\ref{algo5}.

\begin{algorithm}
	\caption {URPCF$_K$}
	\hspace*{0.02in}\raggedright{\bf Input:} A graph $G=(V,E)$ with edge cost function $w$ and vertex penalty function $\pi$, and a positive integer $K$.
	
	\hspace*{0.02in}\raggedright{\bf Output:} A forest $\mathcal{F}_{UR}$ with $K$ connected components.
	
	\begin{algorithmic}[1]
		\STATE $\mathcal{C}\leftarrow$ Rootless-Growth($G$);
		\STATE Construct an auxiliary tree $T_{r'}^{aux}$ from $\mathcal{C}$;
		\STATE $\mathcal{F}_{UR}\leftarrow f(r',0,0,K)$ computed by Rootless-Prune($T_{r'}^{aux},K$);
		\RETURN $\mathcal{F}_{UR}$.
	\end{algorithmic}\label{algo5}
\end{algorithm}

\begin{theorem}\label{theo4}
	{\rm	Algorithm~\ref{algo5} is a 2-LMP for URPCF$_K$ with time complexity $O(n^2K^2+nm)$.}
\end{theorem}

\begin{proof}
	The performance ratio is guaranteed by the correctness of \eqref{eq0104-1}. For time complexity, the procedure Rootless-Growth takes time $O(nm)$. The most time-consuming part in Rootless-Prune is the double dynamic programming. By Corollary~\ref{coro0612-3} and \ref{coro0612-4}, computing one $f(u,p_u,t_u,k_u)$ takes time $O(d_u^2k_u^2)=O(d_u^2K^2)$, where $d_u$ is number of children of $u$ in tree $T_{r'}$. Note that $\sum_{u\in V(T_{r'})}d_u=n-1$. So the time complexity of Rootless-Prune is
	$$
	O\left(K^2\sum_{u\in V(T_{r'})}d(u)^2\right)=O\left(K^2\left(\sum_{u\in V(T_{r'})}d(u)\right)^2\right)=O(n^2K^2).
	$$
	Hence, the total running time of Algorithm~\ref{algo5} is $O(n^2K^2+nm)$.
\end{proof}

\section{Application: A 5-LMP for PCMinSSC}\label{sec003}
In this section, we present a 5-LMP algorithm for PCMinSSC. The algorithm consists of two steps. The first step is {\em grouping}: for each guessed $K$ $(K=1,2,\ldots,n)$, compute a forest $F_K$ with $K$ components for the URPCF$_K$ problem by calling Algorithm~\ref{algo5}. The vertices within the same component $T$ are grouped together, and will be sweep-covered by a shared set of mobile sensors cooperatively. The second step is {\em allocation}, for each group, if it contains exactly one vertex $v$, then one mobile sensor is stationed at $v$. Otherwise construct a cycle $C$ by doubling the edges of $T$ to form an Euler tour and then short-cutting this Euler tour. Because the graph is metric, by triangle inequality, we have
\begin{align}\label{eq1017-1}
	w(C)\leq 2w(T)
\end{align}
(see Chapter 3 of \cite{Vazirani} for the short cutting method). Then $\lceil\frac{w(L)}{at}\rceil$ mobile sensors are uniformly deployed along this cycle and move in the same direction, forming a sweep-cover for those vertices in this group. The final solution is the best one among all guesses of $K$.

\begin{algorithm}[h]
	\caption{$5$-LMP for PCMinSSC}
	\hspace*{0.02in}\raggedright{\bf Input:} An edge weighted graph $G=(V,E,w)$.
	
	\hspace*{0.02in}\raggedright{\bf Output:} Allocation of a set of mobile sensors.
	
	\begin{algorithmic}[1]
		\FOR {$K= 1,\ldots,n$}
		\STATE $F_K\leftarrow$ $K$ vertex-disjoint trees $T_1^{(K)},\ldots,T_K^{(K)}$ computed by Algorithm~\ref{algo5}.
		\STATE $\mathcal C_K\leftarrow$ $\{C_1^{(K)},\ldots,C_k^{(K)}\}$, where $C_i^{(K)}$ is the cycle obtained from $T_i^{(K)}$ by the short-cutting method.
		\STATE $n_K=\sum_{i:|V(T_i^{(K)})|\geq 2}\left\lceil w(C_i^{(K)})/at \right\rceil+|\{i\colon |V(T_i^{(K)})|=1\}|$;
		\ENDFOR
		\STATE $\hat K=\arg\min_{K\in\{1,\ldots,n\}}\{c\cdot n_K+5\sum_{v \in V\setminus V(\mathcal P_K)}\pi(v)\}$;
		\FOR{$i=1,\ldots,\hat K$}
		\STATE If $|V(C_i^{(\hat K)})|\geq 2$, then evenly allocate $\left\lceil w(C_i^{(\hat K)})/at\right\rceil$ mobile sensors on cycle $C_i^{(\hat K)}$, otherwise, allocate one mobile sensor on $C_i^{(\hat K)}$.
		\ENDFOR
		\RETURN  the above deployment.
	\end{algorithmic}\label{algo6}
\end{algorithm}

To analyze the performance of Algorithm \ref{algo6}, a crucial insight is that the trajectories of the mobile sensors within an optimal solution over the time span [0, t] can be trimmed into a collection of vertex-disjoint trees. The proof method aligns with that presented in \cite{Liang-PCMSSC-TCS}, with the distinction that \cite{Liang-PCMSSC-TCS} relies heavily on predetermined roots, whereas this paper does not. In particular, a {\em scaling technique} used in \cite{Liang-PCMSSC-TCS} also plays a crucial role here: it is assumed
\begin{equation}\label{eq1017-2}
	\frac{at}{c}=\frac45.
\end{equation}
We can make this assumption without loss of generality, as we have the flexibility to scale both parameter $\pi$ and coefficient $c$ simultaneously. This scaling operation does not alter the solution and maintains the same approximation ratio. For the sake of completeness in this paper, we provide a detailed proof in the following.

\begin{theorem}
	{\rm	Algorithm~\ref{algo6} is a 5-LMP for PCMinSSC, which can be executed in time $O(n^5)$, where $n$ is the number of vertices.}
\end{theorem}
\begin{proof}
	Suppose that an optimal solution to a PCMinSSC instance uses $K^*$ mobile sensors and denote by $Z^*$ the set of vertices that are not sweep-covered.
	
	{\em Claim.} There exists $K^*$ vertex-disjoint trees $F^*_{K^*}$ spanning $V\setminus Z^*$ such that $$w(F^*_{K^*})\leq k^*at.$$
	
	Note that for each of the $K^*$ mobile sensors, its route in time interval $[0,t]$ is a walk, and the union of these walks, denoted by $W$, spans $V\setminus Z^*$. Deleting some edges if necessary, $W$ can be modified into a forest $F^*_{K^*}$ containing $K^*$ connected components, still spanning $V\setminus Z^*$. Note that $w(W)\leq K^*at$. Therefore,
	\begin{align}\label{eq1017-3}
		w(F^*_{K^*})\leq w(W)\leq K^*at.
	\end{align}
	The claim is proved.
	
	When $K^*$ is guessed in line 1 of Algorithm \ref{algo6}, a $K^*$-forest $F_{K^*}$ is obtained by the 2-LMP algorithm Algorithm~\ref{algo5}. So
	\begin{align}\label{equa0712-1}
		w(F_{K^*})+ 2\pi(V\setminus V(F_{K^*}))\leq 2\left(w(F^*_{K^*})+\pi(V\setminus V(F^*_{K^*}))\right)\leq 2\left(K^*at+\pi(Z^*)\right).
	\end{align}
	
	Denote by $\mathcal S_{\mathcal A}$ the set of mobile sensors computed by Algorithm~\ref{algo6}. The number of mobile sensors is upper bounded by
	$$
	|\mathcal S_{\mathcal A}|\leq \sum_{i=1}^{K^*} \left(\frac{w(C^{(K^*)}_i)}{at}+1\right)
	\leq \sum_{i=1}^{K^*} \left(\frac{2w(T^{(K^*)}_i)}{at}+1\right)=\left(\frac{2w(F_{K^*})}{at}+K^* \right),
	$$
	where the first inequality holds because the number of sensors to sweep-cover those target points grouped by $C^{(K^*)}_i$ is $\max\left\{\left\lceil w(C_i^{(\hat K)})/at\right\rceil,1\right\}\leq w(C^{(K^*)}_i)/at+1$, and the second inequality holds because of \eqref{eq1017-1} and line 3 of Algorithm~\ref{algo6}. Then making use of \eqref{eq1017-2} and \eqref{equa0712-1},
	the sensor-plus-penalty value is upper bounded by
	$$
	c\cdot|\mathcal S_{\mathcal A}|+5\pi(V\setminus V(F_{K^*}))\leq c\cdot\left(\frac{2w(F_{K^*})}{at}+K^* \right)+5\pi(V\setminus V(F_{K^*}))
	\leq 5(c\cdot K^*+\pi(Z^*)).
	$$
	
	As for the running time, grouping the vertices using Algorithm~\ref{algo5} is the most time-consuming part, which takes time $O(n^2K^2+mn)$ by Theorem~\ref{theo4}. Considering the number of guesses, the total time is $\sum_{K=1}^n O(n^2K^2+mn)=O(n^5)$.
\end{proof}

\section{Conclusion and Future Work}\label{sec004}
In this paper, we designed a 2-LMP algorithm for URPCF$_K$ with time complexity $O(n^2K^2+nm)$, based on which we obtain a 5-LMP for PCMinSSC. The algorithm for URPCF$_K$ consists of a rootless-growth step and a rootless-prune step, the former step is realized by treating each vertex equally during the growth and the latter step makes use of a dual problem of finding a forest with the maximum net worth value, which is realized by a double dynamic programming, where two (sub-)dynamic programming algorithms are embedded in a (main-)dynamic programming algorithm.

The performance analysis uses a rooted analog as a bridge: for an instance $I$ of URPCF$_K$, there is a rooted analog $I_R$ achieving the same optimal value as $I$; for $I_R$, we design a 2-LMP algorithm $\mathcal{A}_2$; for $I$, we design an algorithm $\mathcal{A}_1$; by proving that the objective value of $\mathcal{A}_1$ is no worse than that of $\mathcal{A}_2$, we get $\mathcal{A}_1$ is 2-LMP for URPCF$_K$. Such a design and analysis method establishes a relation between rooted and rootless versions without guessing exponential number of roots. We believe that such a method can be applied to some other similar network design problems.

In fact, we have tried to extend the method to deal with the {\em prize-collecting traveling multi-salesmen problem} ($m$-PCTSP). Given a metric graph $G=(V,E,\pi,w)$ as in PCMinSSC, the goal of $m$-PCTSP is to find $m$ cycles to minimize the cost on cycles plus the penalty of the vertices not in cycles. Combining our method with the argument in \cite{Goeman1995} (which extends the 2-approximation algorithm of PCST to a 2-approximation algorithm of PCTSP), we can obtain the following theorem.
\begin{theorem}
	{\rm There exists an $2$-LMP algorithm for $m$-PCTSP.}
\end{theorem}

Note that both the {\em minimum spanning tree spanning $k$ vertices} ($k$-MST) problem (which asks for a minimum tree spanning at least $k$ vertices) and  the {\em budgeted prize-collecting tree} (BPCT) problem (which asks for a tree spanning as many vertices as possible within a budget) have 2-approximations. Do their unrooted forest versions also admit the same approximation ratio? These are interesting topics to be further explored.

\section*{Acknowledgment}
This research is supported by NSFC (U20A2068).

\bibliographystyle{plain}
\bibliography{PCSC}

\end{document}